\documentclass[conference]{IEEEtran}
\IEEEoverridecommandlockouts

\usepackage{times}

\usepackage{graphicx}

\usepackage{multirow}
\usepackage{textcomp}

\usepackage{float}
\usepackage{subfig}

\usepackage{amssymb}
\usepackage{amsthm}

\usepackage[cmex10]{amsmath}

\usepackage{tikz}
\usetikzlibrary{arrows,decorations,positioning,backgrounds,fit,shapes,matrix,shapes.callouts,decorations.text,automata,decorations.pathreplacing}
\usepackage{ifthen}
\usepackage{algorithm}
\usepackage{algorithmic}

\usepackage{comment}

\newcommand{\An}{\operatorname{LCA}}
\def\A{\mathcal{A}}

\def\d{\mathrm{d}}

\newcommand{\ignore}[1]{}

\newtheorem*{theorem6}{Theorem 6}
\newtheorem{theorem}{Theorem}
\newtheorem{lemma}{Lemma}

\newtheorem{claim}{Claim}

\newtheorem{definition}{Definition}

\newtheorem*{theorem*}{Theorem}

\newcommand{\Obst}{\textsc{Obst}}

\newcommand{\Cost}{\emph{Cost}}

\newcommand{\DS}{\texttt{DS}}

\newcommand{\STAT}{\textsc{Stat}}

\def\d{\mathrm{d}}

\def\A{\mathcal{A}}
\def\SP{\textsc{ST}}
\def\DS{\textsc{DS}}

\begin{document}


\title{OBST: A Self-Adjusting Peer-to-Peer Overlay\\Based on Multiple BSTs}

\author{
Chen Avin$^1$, Michael Borokhovich$^{1,*}$, Stefan Schmid$^2$\\
{\small $^1$ Ben Gurion University, Beersheva, Israel;~~ $^2$ TU Berlin \& T-Labs, Berlin, Germany}\\
{\small \texttt{\{avin,borokhovich\}@cse.bgu.ac.il}; \texttt{stefan@net.t-labs.tu-berlin.de}}\\
\thanks{${}^{*}$Michael Borokhovich was supported in part by the Israel Science Foundation (grant 894/09).}
}


\date{}

\maketitle

\begin{abstract}
The design of scalable and robust overlay topologies has been a main research subject since
 the very origins of peer-to-peer (p2p) computing. Today, the
 corresponding optimization tradeoffs are fairly well-understood, at least in the static case and from a worst-case perspective.

This paper revisits the peer-to-peer topology design problem from a self-organization perspective. We initiate the study of topologies which
are \emph{optimized to serve the communication demand}, or even self-adjusting as demand changes.
The appeal of this new paradigm lies in the opportunity to be able to go beyond
the lower bounds and limitations imposed by a static, communication-oblivious, topology. For example, the goal of having short routing paths (in terms of hop count) does no longer conflict with the requirement of having low peer degrees.


We propose a simple overlay topology $\Obst(k)$ which is composed of $k$ (rooted and directed) \emph{Binary Search Trees (BSTs)}, where $k$ is a parameter.
We first prove some fundamental bounds on what can and cannot be achieved optimizing a topology towards a \emph{static} communication pattern (a static $\Obst(k)$). In particular, we show that the number of BSTs that constitute the overlay can have a large impact on the routing costs, and that a single additional BST may reduce the amortized
communication costs from $\Omega(\log{n})$ to $O(1)$, where $n$ is the number of peers.
 Subsequently, we discuss a natural self-adjusting extension of $\Obst(k)$, in which frequently communicating partners are ``splayed together''.

\end{abstract}


\setcounter{page}{1}

\section{Introduction}

Classic literature on the design of peer-to-peer (p2p) topologies typically considers the optimization of \emph{static} properties, such as the peer degree or the network diameter in
the worst case. An appealing alternative is to optimize a p2p system (or more generally, a distributed data structure) based on the communication or usage patterns, either statically (based on known traffic statistics) or dynamically, exploiting temporal localities for self-adjustments.

One of the main metrics to evaluate the performance of a self-adjusting
 network is the \emph{amortized cost}: the worst-case communication cost over time and per request.
Splay trees are the most prominent example of the self-adjustment concept in the context of classic data structures: in their seminal work, Sleator and Tarjan~\cite{sleator1985self} proposed self-adjusting binary search trees where popular items or \emph{nodes} are moved closer to the \emph{root} (where the lookups originate), exploiting potential non-uniformity in the access patterns.

\textbf{Our Contributions.}
This paper initiates the study of how to extend the splay tree concepts~\cite{ipdps13stefan,sleator1985self} to multiple trees, in order to design self-adjusting
 p2p \emph{overlays}. Concretely, we propose a \emph{distributed variant} of the splay tree to build the $\Obst$ overlay: in this overlay, frequently communicating partners are located (in the static
 case) or moved (in the dynamic case)
 topologically close(r), without sacrificing local routing benefits: While in a standard \emph{binary search tree (BST)} a request always originates at the root (we will refer to this problem as the \emph{lookup problem}), in the distributed BST variant, \emph{any pair} of nodes in the network can communicate; we will refer to the distributed variant as the \emph{routing problem}.

 The reasons for focusing on BSTs are based on their simplicity and powerful properties: they naturally support local, greedy routing, they are easily self-adjusted, they support join-leave operations in a straight-forward manner, and they require low peer degrees. The main drawback is obviously the weak robustness imposed by the tree structure, and we address this by using multiple trees.


 The proposed $\Obst(k)$ overlay consists of set of \emph{$k$ distributed BSTs}. (See Figure~\ref{fig:mosts} for an example of a $\Obst(2)$.) We first study how
 the communication cost in a static $\Obst(k)$ depends on the number $k$ of BSTs, and give an upper bound which
 shows that the overlay strictly improves with larger $k$. In fact, we will show that in some situations,
 changing from $k$ to $k+1$ BSTs can make a critical difference in the \emph{routing} cost. Interestingly,
 such a drastic effect is not possible on the classical \emph{lookup} operations in a BST.
 This demonstrates that the problem of optimizing \emph{routing} on a BST has some key differences from the \emph{lookup} problem that was, and still is,  extensively researched.

 \begin{figure}[ht]
	\begin{center}%
\includegraphics[width=0.5\columnwidth]{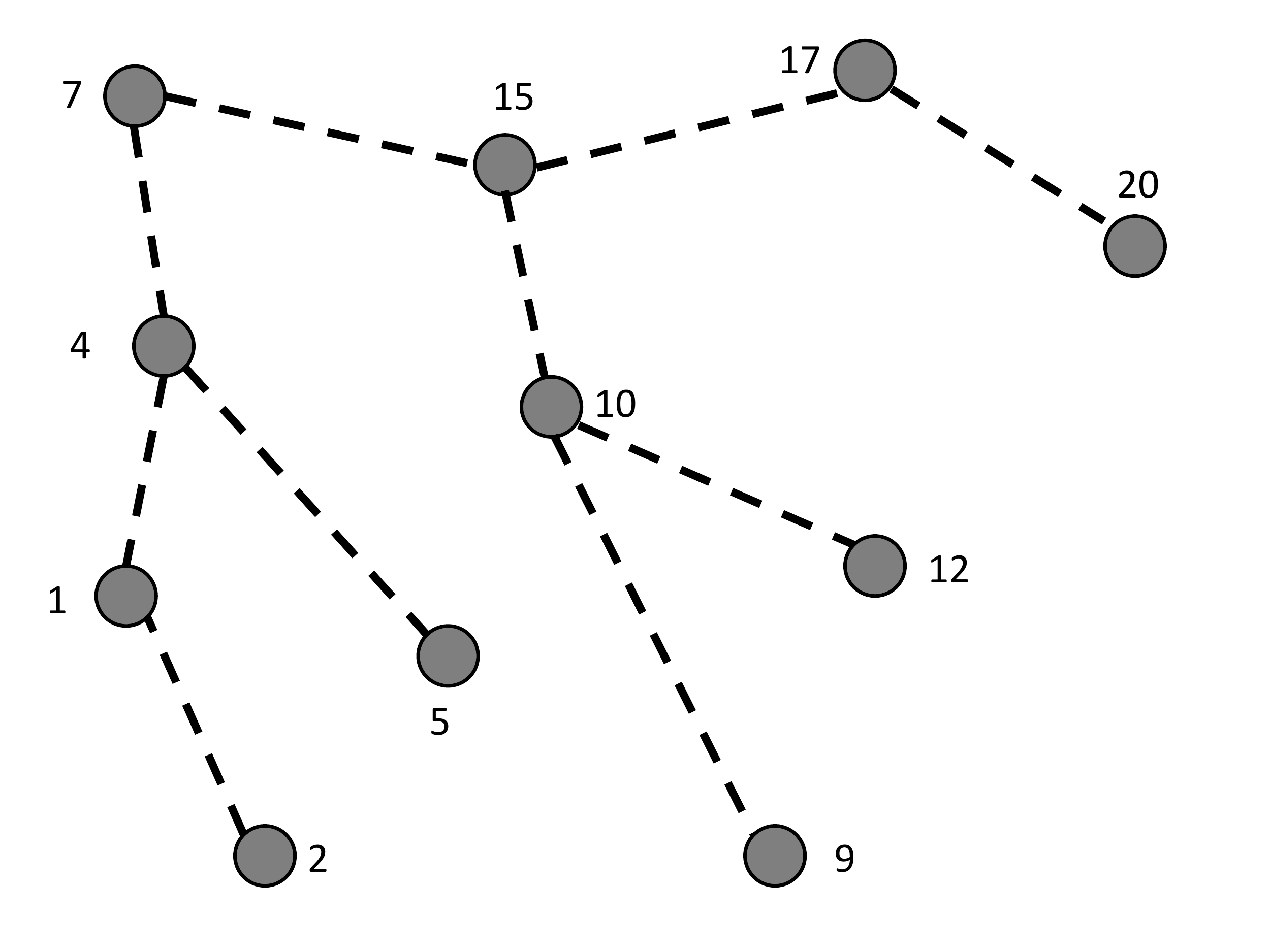}~~\includegraphics[width=0.5\columnwidth]{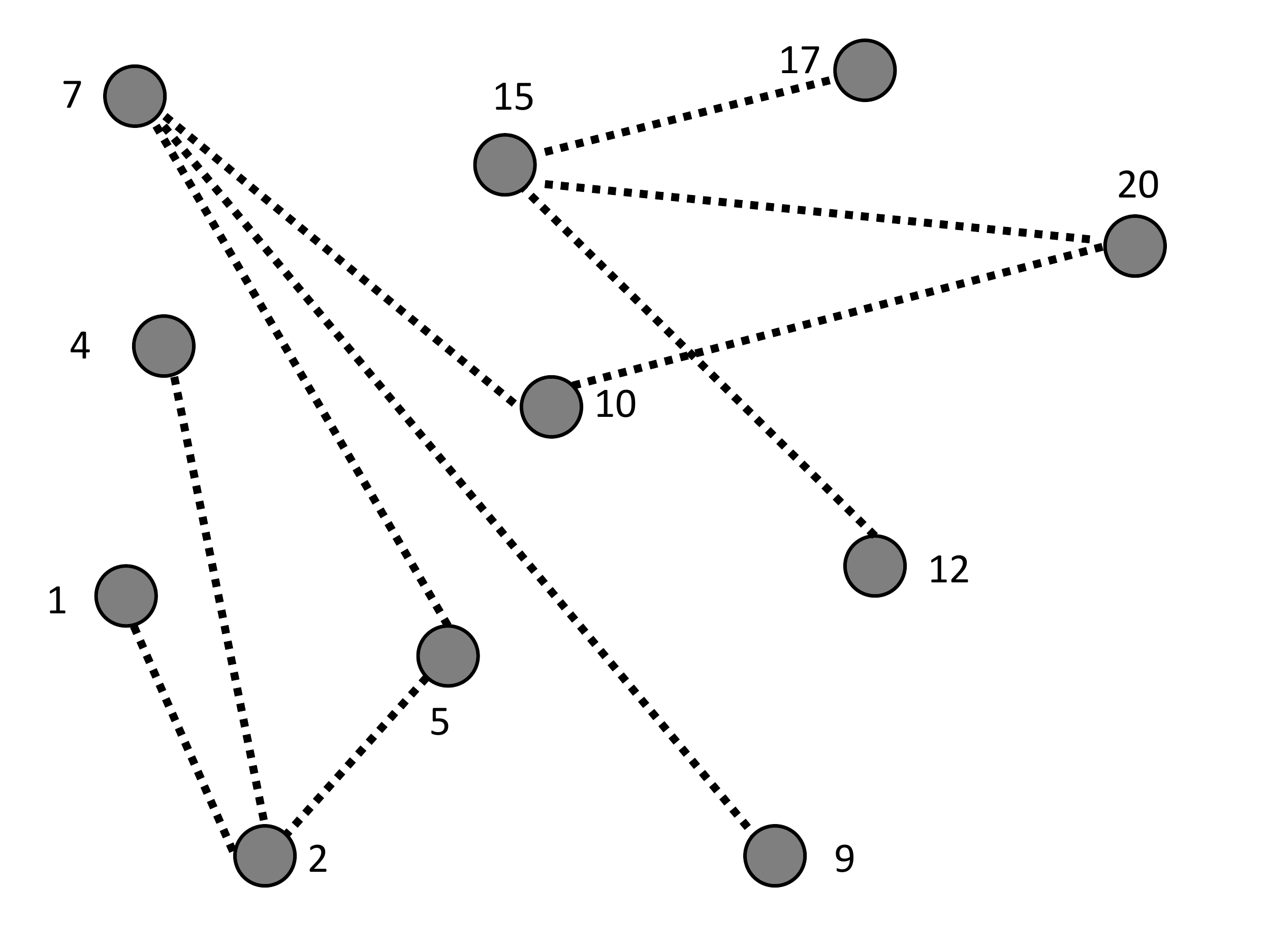}\\%
\vspace{10mm}
\includegraphics[width=0.9\columnwidth]{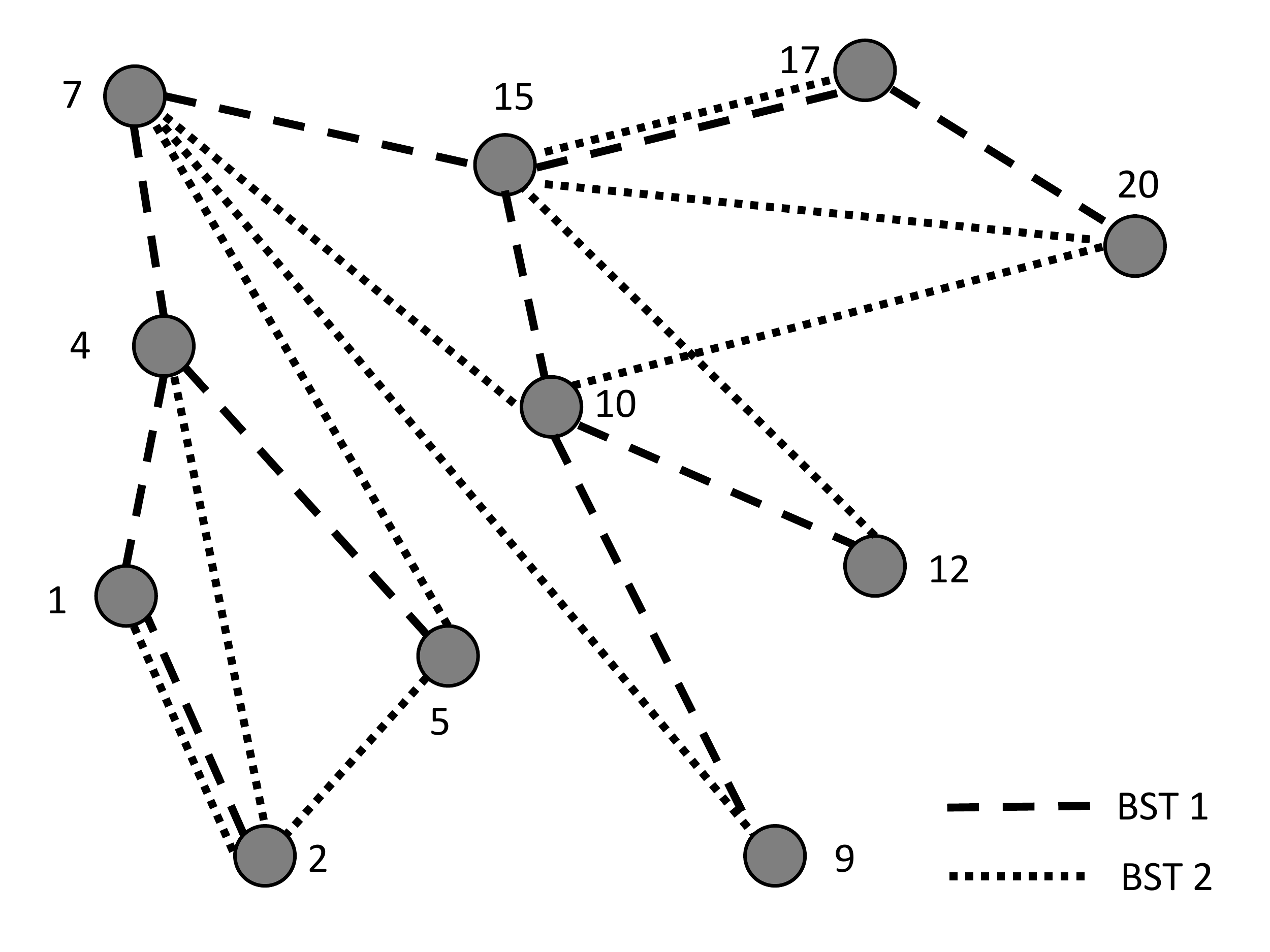}\\%
	\end{center}
\caption{Example of $\Obst(2)$ consisting of two BSTs. Top left: BST 1 (e.g., rooted at
peer $v_7$). Top right: BST 2 (e.g., rooted at peer $v_{10}$). Bottom: combined BSTs.}\label{fig:mosts}
\end{figure}

After studying the static case, we also describe a dynamic and self-adjusting variant of $\Obst$
which is inspired by classic splay trees: communication partners are topologically ``splayed together''.
These splay operations are completely local and hence efficient. We complement our formal analysis by
extensive simulations. These simulation results confirm our theoretical bounds but also reveal some desirable properties in the
time domain (e.g., robustness to failures and churn, or convergence to our static examples).

\textbf{Paper Organization.}
The remainder of this paper is organized as follows. Section~\ref{sec:model} introduces our formal model and the definitions, and Section~\ref{sec:background}
provides the necessary background on binary search and splay trees. We study static $\Obst(k)$ overlays in Section~\ref{sec:staticopt} and dynamic $\Obst(k)$ overlays in Section~\ref{sec:dynopt}. Section~\ref{sec:simulations} reports on our simulation results. The routing model is compared to classic lookup model in Section~\ref{sec:discussion}. After reviewing related work in Section~\ref{sec:relwork}, we conclude our paper in Section~\ref{sec:conclusion}. Finally, in the appendix, the existence and limitations of optimal
overlays are discussed.

\section{Model and Definitions}\label{sec:model}

We describe the p2p overlay network as a graph $\mathcal{H}=(V,E)$ where $V=\{v_1,\ldots,v_n\}$ is the set of peers and $E$ represents their connections.
For simplicity, we will refer by $v_i$ both to the corresponding peer as well as the peer's (unique) identifier; sometimes, we will simply write $i$ instead of $v_i$. Moreover, we
will focus on bidirected overlays, i.e., we will ensure that if a peer $v_1\in V$ is connected to another peer $v_2\in V$, denoted by $(v_1,v_2)$,
then also $v_2$ is connected to $v_1$ (i.e., $(v_2,v_1)$). Sometimes we will refer to the two bidirected edges $(v_1,v_2)$ and $(v_2,v_1)$ simply by
$\{v_1,v_2\}$.

We will assume that peers communicate according to a certain pattern. This pattern may be \emph{static} in the sense that
it follows a certain probability distribution; or it may be \emph{dynamic} and change arbitrarily over time. Static communication patterns
may conveniently be represented as a weighted directed graph $\mathcal{G}=(V,E)$: any peer pair $(v_1,v_2)$ communicating with a non-zero probability
is connected in the graph $\mathcal{G}$.

We will sometimes refer to the sequence of communication events between peers as \emph{communication requests} $\sigma$. In the static case,
 we want the overlay $\mathcal{H}$ be as similar as possible to the communication pattern $\mathcal{G}$ (implied by $\sigma$), in the sense
  that an edge $e\in E(\mathcal{G})$ is represented by a short route in $\mathcal{H}$; this can be seen as a graph embedding problem of
  $\mathcal{G}$ (the ``guest graph'') into $\mathcal{H}$ (the ``host graph'').
In the dynamic setting, the topology $\mathcal{H}$ can be adapted over time depending on $\sigma$. These topological transformations
should be \emph{local}, in the sense that only a few peers and links in a small subgraph are affected.

Our proposed topology $\Obst(k)$ can be described by a simple graph $\mathcal{H}$ which consists of a set of $k$ \emph{binary search trees (BST)}, for some $k>0$.
\begin{definition}[$\Obst(k)$]
Consider a set $\{T_1,T_2,\ldots,T_k\}$ of $k$ BSTs.
$\Obst(k)$ is an overlay over the peer set $V=\left\{1,\ldots,n\right\}$ where
connections are given by the BST edges, i.e., $E=\bigcup_{i=1}^k E(T_i)$.
\end{definition}

Our topological transformations to adapt the $\Obst(k)$ are \emph{rotations} over individual BSTs: minimal and local
transformations that preserve a BST. Informally, a rotation in a sorted binary search tree changes the
local order of three connected nodes, while keeping subtrees intact. Note that
it is possible to transform any binary search tree into any other
binary search tree by a sequence of local transformations (e.g., by induction over
the subtree roots).


Let $\sigma =
(\sigma_0, \sigma_1 \dots )$ be a sequence of $m$ \emph{requests}. Each request $\sigma_t =(u,v)$ is a pair of a source peer and a destination peer. Let $\A$ be an algorithm that given the request $\sigma_t$ and the graph $\mathcal{H}_t$ at time $t$, transforms the current graph (via local
transformations) to $\mathcal{H}_{t+1}$ at time $t+1$. We will use $\STAT$ to refer to an any static (i.e., non-adjusting) ``algorithm'' which does not
change the communication network over time; however, $\STAT$ is initially allowed to choose an overlay
which reflects the statistical communication pattern.

The cost of the network transformations at time $t$ are denoted by $\rho(\A,
\mathcal{H}_t, \sigma_t)$ and capture the number of rotations performed to
change $\mathcal{H}_t$ to $\mathcal{H}_{t+1}$; when $\A$ is clear from the context, we
will simply write $\rho_t$. We denote with $\d_{\mathcal{H}}(\cdot)$ the distance function between nodes in $\mathcal{H}$, i.e., for two nodes $v, u
\in V$ we define $\d_{\mathcal{H}}(u,v)$ to be the number of edges of a
\emph{shortest} path between $u$ and $v$ in $\mathcal{H}$. (The subscript $\mathcal{H}$ is optional if clear from the context.)
Note that for a BST $T$, the shortest path between $u$ and $v$ is unique and can be found and routed locally via a greedy algorithms.

For a given sequence of communication requests, the cost for an
algorithm is given by the number of transformations and
the distance of the communication requests.
Formally, we will make use of the following standard definitions (see also~\cite{ipdps13stefan}).
\begin{definition}[\bf \emph{Average and Amortized Cost}]
For an algorithm $\A$ and given an initial network $\mathcal{H}_0$ with node distance function
$\d(\cdot)$ and a sequence $\sigma=(\sigma_0,
\sigma_1 \dots \sigma_{m-1})$ of communication requests over time,
we define the \emph{(average) cost} of $\A$ as:
$\Cost( \A, \mathcal{H}_0, \sigma) = $ $\frac{1}{m}\sum_{t=0}^{m} (\d_{\mathcal{H}_t}(\sigma_t)+1 $ $+
\rho_t)
$
The \emph{amortized
cost} of $\A$ is defined as the worst possible cost of $\A$, i.e.,
$\max_{\mathcal{H}_0,\sigma}\Cost( \A, \mathcal{H}_0, \sigma)$.
\end{definition}

One may consider two different routing models on $\Obst$. (We will review how to do local routing in BSTs in Section~\ref{sec:background}.) In the first model, two peers will always communicate along a \emph{single} BST: one which minimizes
the hop length; the best BST may be found, e.g., via a probe message along the trees: the first response is taken. In the second model, we
allow routes to cross different BSTs, and take the globally shortest path; this can be achieved, e.g., by using a standard routing protocol
(e.g., distance vector) in the background. In the following, if not stated differently, we will focus on the first model, which is more
conservative in the sense that it yields higher costs.

\section{Background on BSTs and Splay Trees}\label{sec:background}

The following facts are useful in the remainder of this paper.
Theorem~\ref{static_lookup_lower_single} bounds the lookup cost
in an optimal binary search tree under a given \emph{lookup} sequence $\sigma$: a sequence of requests all originating
from the root of the tree.
\begin{theorem}[\cite{mehlhorn1975nearly}] \label{static_lookup_lower_single}
Given $\sigma$, for any (optimal) BST $T$, the amortized cost is at least
\begin{align}\label{eq:lower}
\Cost( \STAT, T, \sigma) \ge \frac{1}{\log 3} H(\hat{Y})
\end{align}
\noindent where $\hat{Y}(\sigma)$ is the empirical measure of the frequency distribution of $\sigma$ and  $H(\hat{Y})$ is its empirical entropy.
\end{theorem}

Knuth \cite{knuth1971optimum} fist gave an algorithm to find optimal BST, but Mehlhorn \cite{mehlhorn1975nearly} proved that a simple greedy algorithm is near optimal with an explicit bound:
\begin{theorem}[\cite{mehlhorn1975nearly}] \label{static_lookup_upper_single}
Given $\sigma$, there is a BST, $T_{\mathrm{balanced}}$ that can be computed using a balancing argument and has the amortized cost that is at most
\begin{align}\label{eq:upper_single}
\Cost( \STAT, T_{\mathrm{balanced}}, \sigma)
\le 2+ \frac{H(\hat{Y})}{1-\log (\sqrt{5}-1)}
\end{align}
\noindent where $\hat{Y}(\sigma)$ is the empirical measure of the frequency distribution of $\sigma$ and  $H(\hat{Y})$ is its empirical entropy.
\end{theorem}

Sleator and Tarjan were able to show that
splay trees, a self-adjusting BST with an algorithm which we denote $\SP$, yields the same amortized cost as an optimal binary search tree.
\begin{theorem}[Static Optimality Theorem \cite{sleator1985self} - rephrased]
Let  $\sigma$ be a sequence of lookup requests
where each item is requested at least once, then for any initial tree $T$
$\Cost(\SP, T, \sigma) = O(H(\hat{Y}))$ where $H(\hat{Y})$ is the empirical entropy of $\sigma$.
\end{theorem}

In \cite{ipdps13stefan}, Avin et al.~proposed a \emph{single} dynamic splay BST for routing, and a double splay algorithm
we denote as $\DS$. For the single tree case and any initial tree $T$ the authors proved the following lower bound for $\STAT$:
 $$\Cost( \STAT, T, \sigma) = \Omega(H(\hat{Y} | \hat{X}) +  H(\hat{X} | \hat{Y}))$$
and the following upper bound for $\DS$:
 $$\Cost( \DS, T, \sigma) = O(H(\hat{X}) + H(\hat{Y}))$$
where $\hat{X}$ and $\hat{Y}$ are the empirical measures of the frequency distribution of the sources and destinations from $\sigma$, respectively and $H$ is the entropy function.

Finally, it is easy to see that BSTs support simple and local routing. For completeness, let us review the proof from~\cite{ipdps13stefan} (adapted to our terminology).
\begin{claim}\label{clm:locrouting}
BSTs support local routing.
\end{claim}
\begin{proof}
Let us regard each peer $u$ in the BST $T$ as the root
of a (possibly empty) subtree $T(u)$. Then, a node $u$ simply needs to store the smallest identifier $u'$ and the largest identifier $u''$ currently present in $T(u)$.
This information can easily be maintained, even under the topological transformations performed by our algorithms.
When $u$ receives a packet for destination address $v$, it will forward it as follows: (1) if $u=v$, the packet reached its destination;
(2) if $u'\leq v \leq u$, the packet is forwarded to the left child and similarly, if $u\leq v \leq u''$, it is forwarded to the right child; (3)
otherwise, the packet is forwarded to $u$'s parent.
\end{proof}

%

\section{Static $\Obst(k)$ Optimization for P2P}\label{sec:staticopt}

We will first study static overlay networks which are optimized towards a request distribution given beforehand.
The number of BSTs $k$ is given together with the sequence of communication requests $\sigma=(\sigma_0,\sigma_1,\ldots)$.
The goal is to find the optimal $\Obst(k)$ to minimize $\Cost( \STAT, \Obst(k), \sigma)$.

In \cite{ipdps13stefan} it is was proved that for any $\sigma$, the optimal $\Obst(1)$ can be found in polynomial time.
Here we first provide a new upper bound for the optimal $\Obst(k)$ and show how it can improve with $k$.

For communication requests $\sigma$ let $x_i(\sigma)$ (or for short $x_i$) be the frequency of $v_i$ as a \emph{source} in  $\sigma$, similarly let $y_i$ be the frequency of $v_i$ as a \emph{destination} and $f_{ij}$ be the frequency of the request $(v_i, v_j)$ in $\sigma$.
Define $z_i = (x_i + y_i)/2$ and note that by definition $\sum_1^n z_i = 1$. Let $\hat{Z}$ be a random variable (r.v.) with a probability distribution defined by the $z_i's$.
For any $k$ partition of the requests in $\sigma$ into disjoint sets $S_1, S_2, \dots , S_k$, let $\alpha_1,\alpha_2,\ldots,\alpha_k$ be the frequency
measure of the partition, i.e.,
$\alpha_i = \sum_{(i,j) \in S_i} f_{ij}$.

First we can prove a new bound on the optimal static $\Obst(1)$:
\begin{theorem} \label{static_lookup_upper_1_set}
Given $\sigma$, there exists a $\Obst(1)$ such that:
\begin{align}
\Cost( \STAT, \Obst(1), \sigma)
\le 4+ \frac{2H(\hat{Z})}{1-\log (\sqrt{5}-1)} \nonumber
\end{align}

\noindent where $H(\hat{Z})$ is the entropy of $\hat{Z}$ as defined earlier.
\end{theorem}
\begin{proof}
The result follows from Theorem \ref{static_lookup_upper_single} with some modifications.
Consider a tree $T$ and let $\ell_i$ denote the distance of node $i$ from the root. We will assume the following non-optimal strategy:
each request $(i,j)$ is first routed from $i$ to the root and then from the root to $j$.
The amortized cost of $\sigma$ can be written as the sum of $i$ routing to and from the root
 \begin{align}
\sum_{i=1}^n x_i \ell_i + \sum_{i=1}^n y_i \ell_i &= 2\sum_{i=1}^n \frac{(x_i +y_i)}{2} \ell_i =  2\sum_{i=1}^n z_i \ell_i
\end{align}
Now given $z_i$, the problem of finding the tree that minimizes $\sum_{i=1}^n z_i \ell_i$ is exactly the lookup problem of
Theorem \ref{static_lookup_upper_single} and the result follows.
\end{proof}

Consider now the $\Obst(k)$ overlay which consists of $k$ BSTs. Assume again a non-optimal strategy:
we partition $\sigma$ into $k$ disjoint sets of requests $S_1, S_2, \dots , S_k$, and each request is routed on its unique BST.
In each tree we use the previous method, and the messages are routed from the source to the root and from the root to the destination.

We can now prove an upper bound on  $\Obst(k)$  that improves with $k$.

\begin{theorem} \label{static_lookup_upper_k_set}
Given $\sigma$, there exists a $\Obst(k)$ such that:
\begin{footnotesize}
\begin{align}
\Cost( \STAT, \Obst(k), \sigma)
\le 4+ \frac{2H(\hat{Z})-2H(\alpha_1,\alpha_2,\ldots,\alpha_k)}{1-\log (\sqrt{5}-1)} \nonumber
\end{align}
\end{footnotesize}
\noindent where $H(\hat{Z})$ is the entropy of $\hat{Z}$ as defined earlier.
\end{theorem}


%

\begin{proof}
For a subset $S_i$, $1 \le i \le k$, let $\hat{Z}_i$ denote the frequency measure defined as $\hat{Z}$, but limited to the requests in $S_i$.
Now:
\begin{footnotesize}
\begin{align}
&\Cost( \STAT, \Obst(k), \sigma) \le \sum_{i=1}^k \alpha_i \left(4+2\frac{H(\hat{Z}_i)}{1-\log (\sqrt{5}-1)} \right)\\
&=4+ \frac{2}{1-\log (\sqrt{5}-1)}\sum_{i=1}^k \alpha_i H(\hat{Z}_i)\\
&=4+ \frac{2}{1-\log (\sqrt{5}-1)}\left(H(\hat{Y})-H(\alpha_1,\alpha_2,\ldots,\alpha_k)\right)
\end{align}
\end{footnotesize}
where the last step is based on the decomposition property of entropy.
\end{proof}

Note that this approach can yield a cost reduction of up to $\log k$, when the $\alpha_i$ values are equal. The problem of equally partition $\sigma$ into $k$ sets in order to maximize $H(\alpha_1,\alpha_2,\ldots,\alpha_k)$ is NP-complete, since even the partition problem (i.e., $k=2$) and in particular the balanced partition problem (with $k=2$) are NP-complete \cite{jis1979computers}. Interestingly, for those cases, $k=2$,  a pseudo-polynomial time dynamic programming algorithm exist.

The bound in Theorem~\ref{static_lookup_upper_k_set} is conservative in the sense that sometimes, a single additional BST
can reduce the optimal communication cost of $\Obst(k)$ from worst possible (e.g., $\Omega(\log n)$) to a constant cost in $\Obst(k+1)$.
\begin{theorem}\label{thm:2trees-bad}
A single additional BST can reduce the amortized costs from a best possible value of $\Omega(\log{n})$ to $O(1)$.
\end{theorem}
The formal proof appears in the appendix. Essentially, it follows from the two BSTs $T_1=(V,E_1)$ and $T_2=(V,E_2)$ shown in Figure~\ref{fig:twotrees}:
obviously, the two BSTs can be perfectly embedded into $\Obst(2)$ consisting of two BSTs as well. However, embedding the two trees at low cost
in one BST is impossible, since there is a large cut in the identifier space. See the proof for details.

\vspace{6pt}

\vspace{3mm}
\begin{figure}[ht]
	\begin{center}%
\includegraphics[width=0.9\columnwidth]{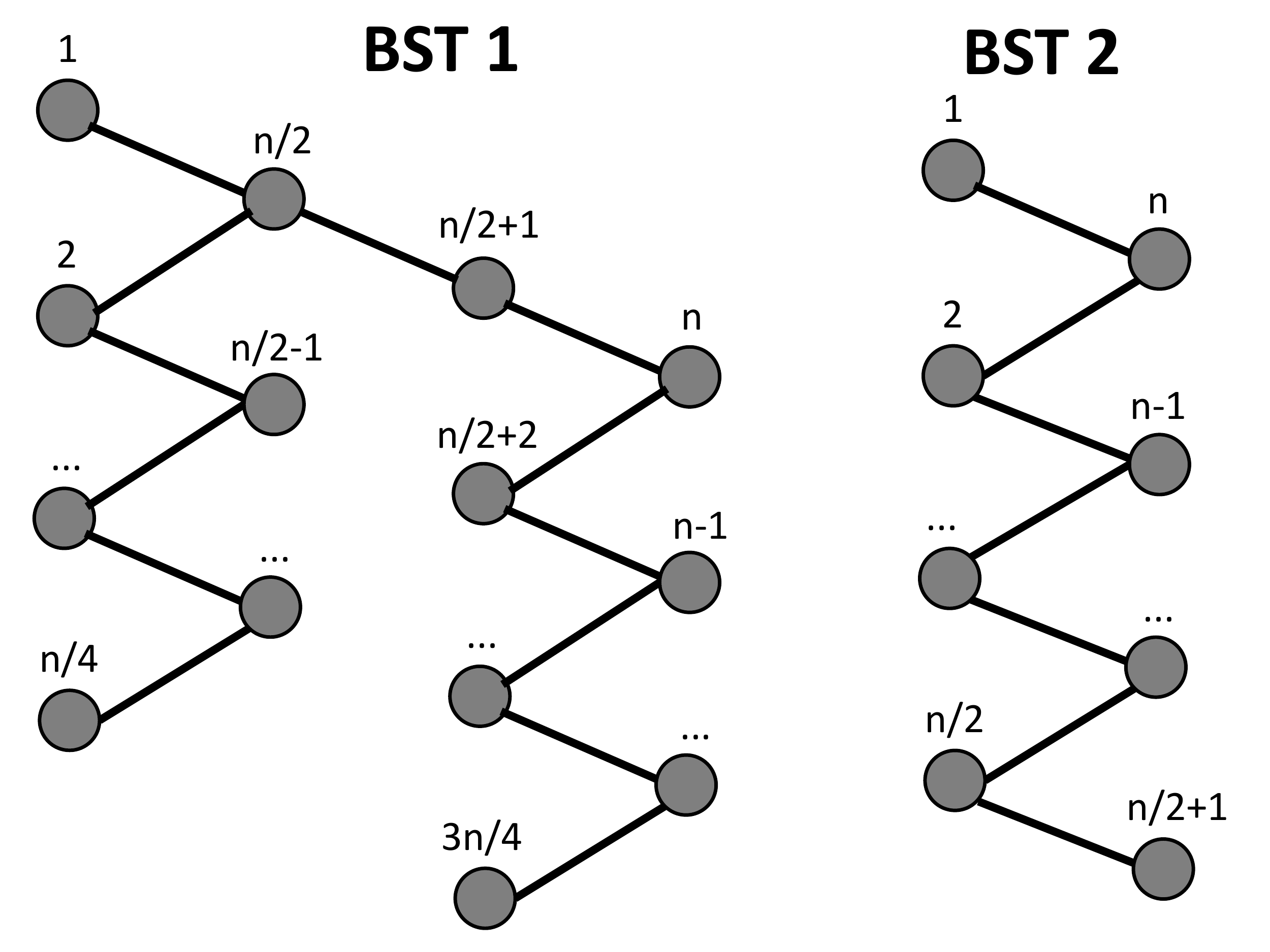}\\%
	\end{center}
\caption{A request sequence $\sigma$ originating from these specific trees can yield
high amortized costs.}\label{fig:twotrees}
\end{figure}
\vspace{3mm}


Interestingly, as we will discuss in Section~\ref{sec:discussion}, such a high benefit from one additional BST is unique to the routing model and does
not exist for classic lookup data structures. Moreover, as we will see in Section~\ref{sec:simulations}, Theorem~\ref{thm:2trees-bad} even holds in a dynamic setting,
i.e., a p2p system can also \emph{converge} to such a bad situation.

\section{Dynamic Self-Adjusting $\Obst(k)$ Overlay}\label{sec:dynopt}

Given our first insights on the performance of static $\Obst(k)$ networks, let us now initiate the discussion
of self-adjusting variants: BSTs which adapt to the demand, i.e., the sequence $\sigma$.

\subsection{Splay Method}

We initialize $\Obst(k)$ as follows: each BST connects \emph{all} peers $V$ as a random and independent binary search tree.

When communication requests occur, BSTs start to adapt. In the following, we will adjust the overlay at each
interaction (``communication event'' or ``request'') of two peers. Of course, in practice, such frequent changes
are undesirable. While our protocol can easily be adapted such that peers only initiate the topological rearrangements after a certain number of interactions
(within a certain time period), in order to keep our model simple, we do not consider these extensions here.

Concretely, we propose a straight-forward \emph{splay} method (inspired from the classical splay trees) to change the $\Obst(k)$:
whenever a peer $u$ communicates with a peer $v$, we perform a distributed splay operation
in \emph{one} of the BSTs, namely in the BST $T$ in which the two communication partners $(u,v)$ are
the topologically closest.

Concretely, upon a communication request $(u,v)$, we determine the BST $T$ (in case multiple
trees yield similar cost, an arbitrary one is taken), as well as the least common ancestor $w$
of $u$ and $v$ in $T$: $w:=\An_{T}(u,v)$. Subsequently, $u$ and $v$ are splayed to the root of the subtree (henceforth denoted by $T(w)$) of $T$ rooted at $w$
(a so-called \emph{double-splay} operation~\cite{ipdps13stefan}).

\vspace{3mm}
\begin{algorithm}[h]
    \caption{Dynamic $\Obst(k)$}
    \label{alg:splaynet}
    \begin{algorithmic}[1]
        \STATE (* upon request $(u,v)$ *)
        \STATE find BST $T\in \Obst$ where $u$ and $v$ are closest;
        \STATE $w:=\An_{T}(u,v)$;
        \STATE $T'$ $:=$ \textbf{splay} $u$ to root of $T(w)$;
	    \STATE \textbf{splay} $v$ to the child of $T'(u)$;
    \end{algorithmic}
\end{algorithm}
\vspace{3mm}

Figure~\ref{fig:sa} gives an example: upon a communication request between peers $v_5$ and $v_{12}$, the two peers are splayed to their least common ancestor, peer $v_7$,
in BST $T_1$.
 \begin{figure}[ht]
	\begin{center}%
\includegraphics[width=0.9\columnwidth]{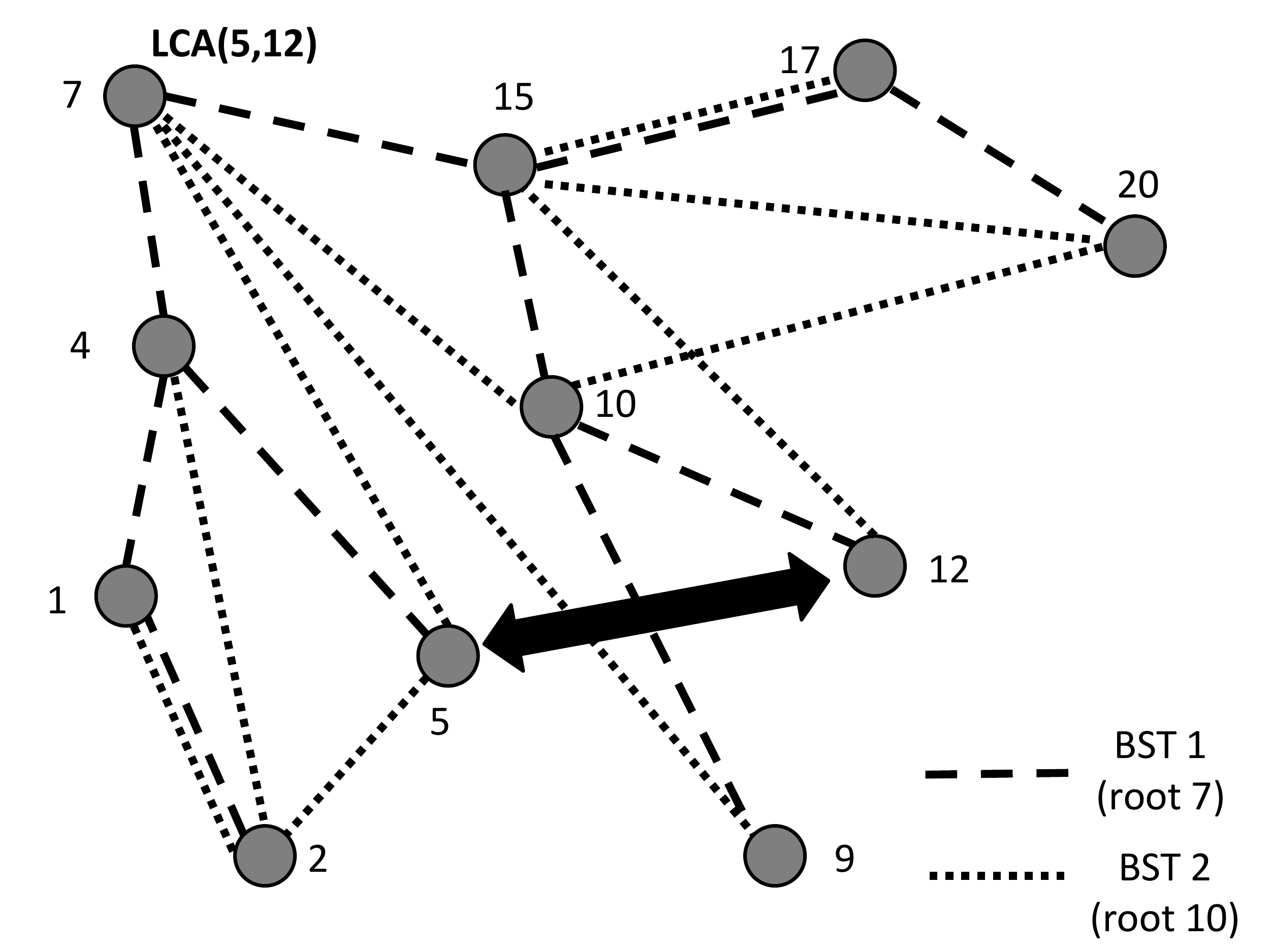}\\%
	\end{center}
\caption{Example for splay operation in $\Obst(2)$ of Figure~\ref{fig:mosts}.}\label{fig:sa}
\end{figure}

\subsection{Joins and Leaves}\label{ssec:joinleave}

Note that $\Obst(k)$ automatically supports joins and leaves of peers.
In order for a peer $v$ to join $\Obst(k)$, it is added as a \emph{leaf} to each BST constituting the overlay (according to the search order).
Similarly for leaving: In order for a peer $v$ to leave, for each BST, $v$ is first swapped with either the rightmost peer of its left subtree (its in-order predecessor) or the leftmost node of its right subtree (its in-order successor); subsequently it is removed. In case of crash failures, a neighboring peer is responsible to detect
that $v$ left and to perform the corresponding operations.

\section{Simulations}\label{sec:simulations}

In order to study the behavior of the self-adjusting $\Obst(k)$, we conducted an extensive simulation study.
In particular, we are interested in how the performance of $\Obst$ depends on the number of BSTs $k$ and the specific communication patterns.

\subsection{Methodology}

We generated the following artificial communication patterns. These patterns were obtained by
first constructing some guest graph $\mathcal{G}$, and then generating $\sigma$ from $\mathcal{G}$.

To model communication patterns, we implemented four guest graphs $\mathcal{G}$:
\\
\noindent
\textbf{(1)} BitTorrent Swarm Connectivity (\textsc{BT}): $\mathcal{G}$ models the p2p connectivity patterns measured by Zhang et al.~\cite{bt-ross}. The
model is taken from~\cite{p2p12stefan} and combines preferential attachment aspects (for swarm popularity) with clustering (for common interest types).
Concretely, the network is divided into a collection of (fully connected) swarms, and the join probability of a peer is proportional to the number of \emph{its neighbors} participating in that swarm.
\\
\noindent
\textbf{(2)} Facebook (\textsc{FB}): $\mathcal{G}$ is a connected subset of the Facebook online social network (obtained from \cite{facebook-viswanath-2009}).
 The graph consists of roughly $63K$ nodes and $800K$ edges. The peer identifiers were chosen according to a breadth first search, starting from the peer with the highest degree (breaking ties at random). For graphs with less then 63k peers ($n<63$k), the subgraph with $n$ peers with the smallest identifiers was extracted.
\\
\noindent
\textbf{(3)} $\textsc{Rnd}(16)$: $\mathcal{G}$ is simply an $\Obst$ with (16) randomly generated BSTs.
\\
\noindent
\textbf{(4)} $\textsc{Bad}(2)$: $\mathcal{G}$ is an $\Obst$ with (two) BSTs generated specifically for the worst routing cost in $\mathcal{H}=\Obst(1)$.

To generate $\sigma$, we used three methods:
\\
\noindent \textbf{(1)} \textsc{Match}: The sequence $\sigma$ is generated from a random maximal matching on $\mathcal{G}$. After each matching edge has been used once,
the next random matching is generated.
\\
\noindent
\textbf{(2)} \textsc{RW-0.5}: The sequence $\sigma$ models a random walked performed on $\mathcal{G}$. Every edge (i.e., request) $(u,v)$ of the random walk is repeated with probability $0.5$; with probability $0.5$, another random request $(v,w)$ is generated.
\\
\noindent
\textbf{(3)} \textsc{RW-1.0}: Like \textsc{RW-0.5}, but without repetitions.

\subsection{Impact of the Number of BSTs}

Let us first study how the routing cost depends on $k$, the number of BSTs in $\Obst(k)$. The routing cost of a communication request $(u,v)$ in $\mathcal{H} = \Obst(k)$ is given by the shortest distance $\d(u,v)$ (among all the $k$ BSTs).
Figure~\ref{fig:fb_pref_ktg16_match} plots the average routing cost for our guest graphs \textsc{BT}, \textsc{FB}, and $\Obst$, under the maximal matching request pattern (\textsc{Match}).
 For each experiment, we generate $n^2$ requests (i.e., $\sigma$ is a sequence of $n^2$ pairs), which is sufficient to explore the performance of $\Obst$ over time. (The variance over different runs is very low.)

We observe that \textsc{BT} typically yields slightly higher costs than \textsc{FB} and especially \textsc{Obst}. As a rule of thumb, doubling $k$ roughly yields a constant additive improvement in the routing cost, in all the scenarios. Interestingly, $\Obst(k)$ is indeed able to perfectly embed $\Obst(k')$ requests with $k'\leq k$:
 $\Obst$ converges to the optimal structure in which every BST of $\mathcal{H}$ serves a specific BST of $\mathcal{G}$.
But even for $k'>k$, $\Obst(k)$ is able to exploit locality and the cost is relatively stable and independent of the size of the p2p system.

\begin{figure*}
\centering
\subfloat[]{\includegraphics[width=0.325\textwidth]{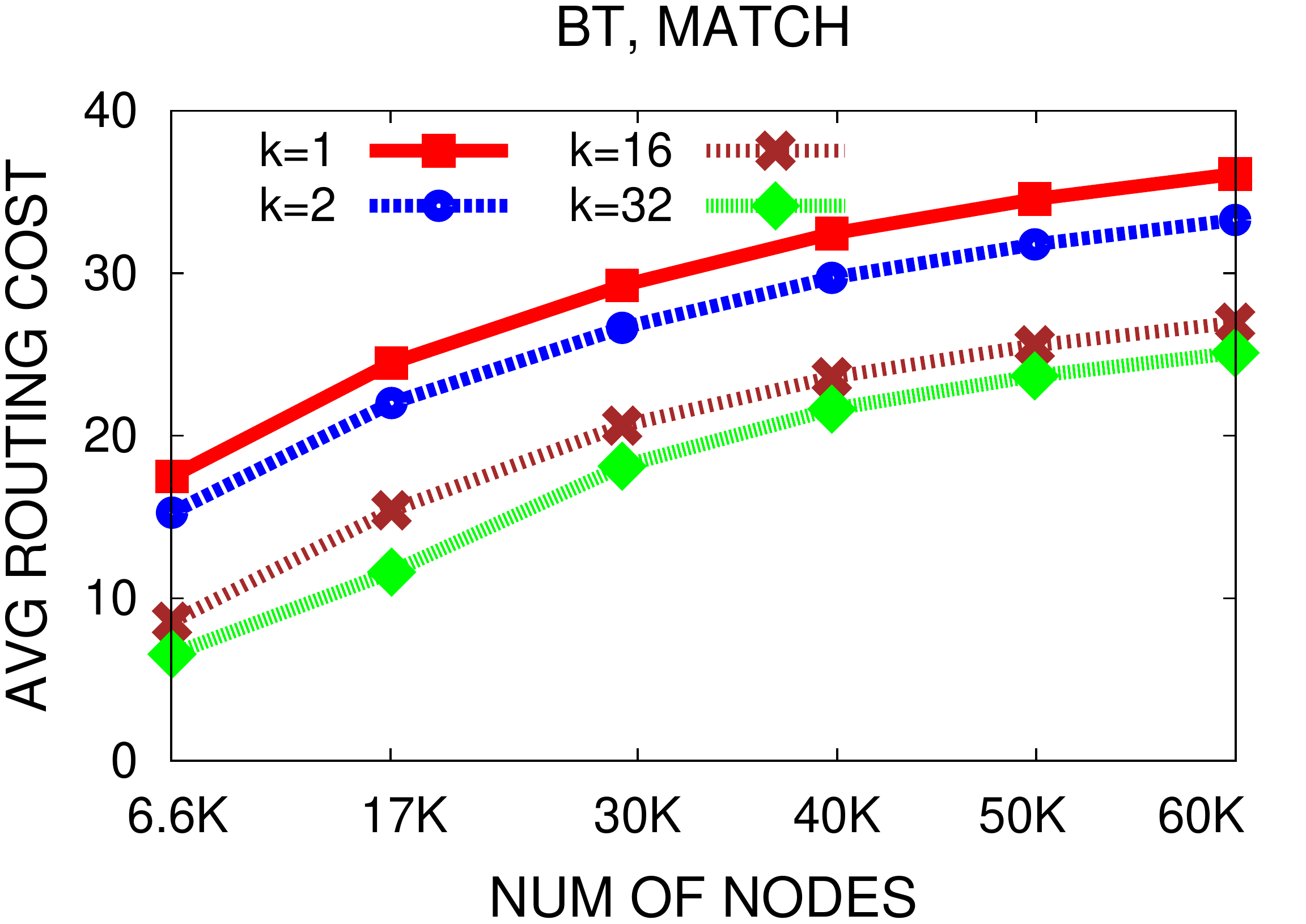}}
\subfloat[]{\includegraphics[width=0.325\textwidth]{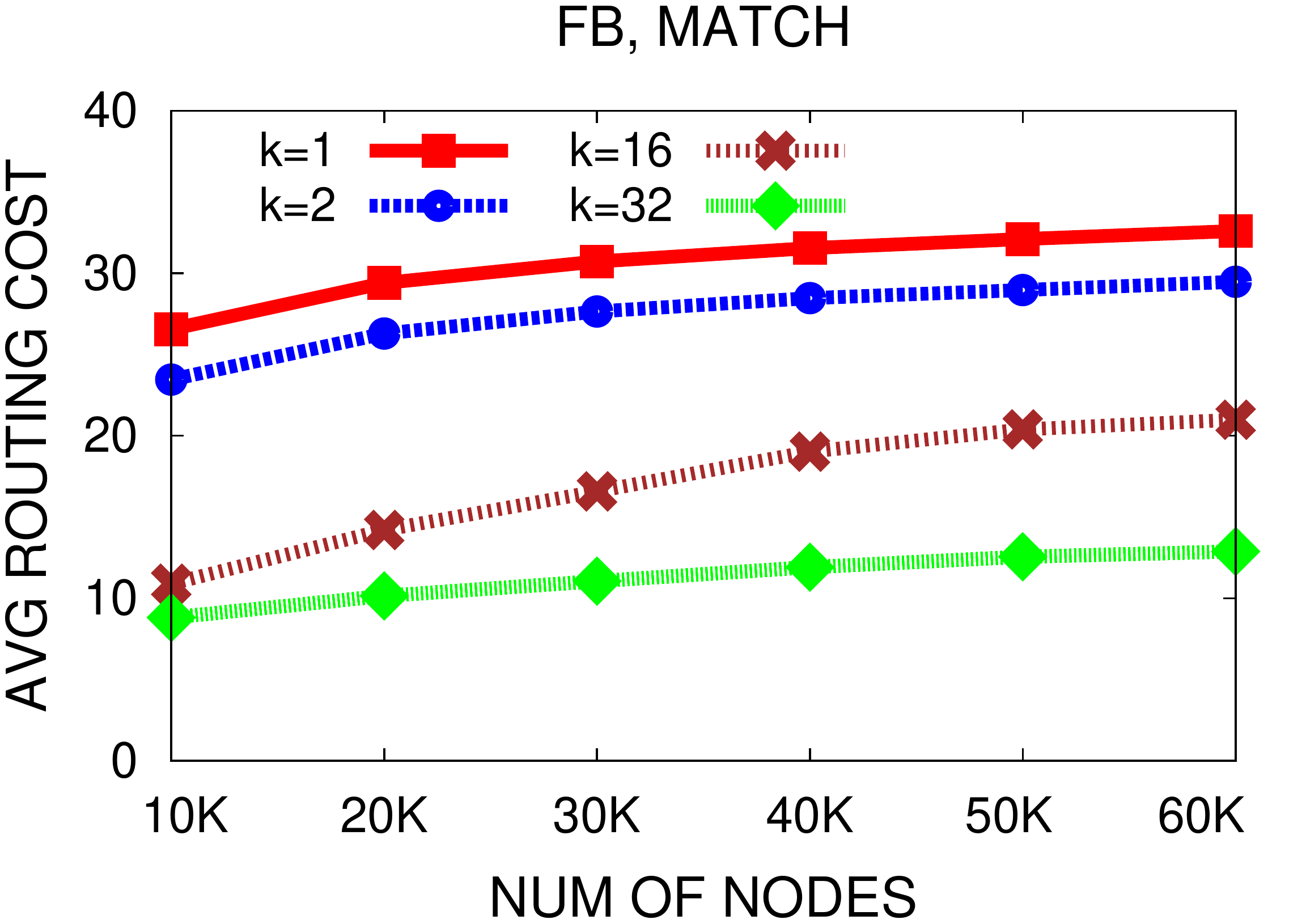}}
\subfloat[]{\includegraphics[width=0.325\textwidth]{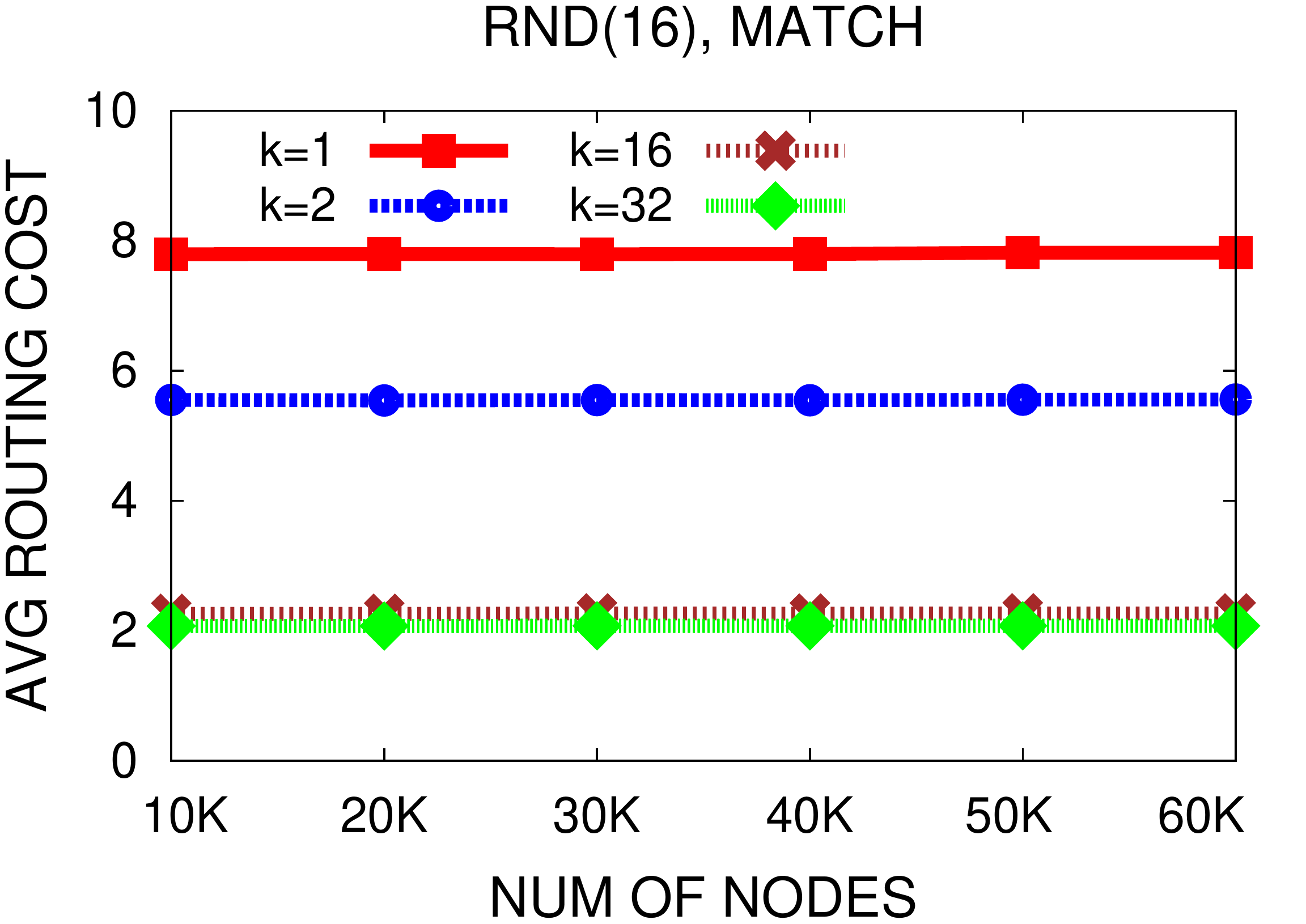}}
\caption{Average routing distance in $\Obst$ as a function of the number of BSTs $k$ and the network size $n$.
The \textsc{Match} method is applied to the graphs from BitTorrent swarms (\textsc{BT}), Facebook (\textsc{FB}),  and $\Obst(16)$.}
\label{fig:fb_pref_ktg16_match}
\end{figure*}

Figure~\ref{fig:worst_requests} shows the routing cost for $\Obst(1)$ and $\Obst(2)$, for a $\textsc{Bad}(2)$ scenario on \textsc{FB}. The figure
confirms Theorem \ref{thm:2trees-bad} and extends it for the dynamic case: $\Obst$ may also converge to a situation as shown to exist in Theorem \ref{thm:2trees-bad},
and a single additional BST improves the routing cost by an order of magnitude. The figure also confirms that multiple BSTs are more useful than in the pure lookup model of Section~\ref{sec:discussion}. Finally, note that unlike Figure~\ref{fig:fb_pref_ktg16_match} (c), the cost is not independent of the network size under this worst-case request pattern.

\begin{figure}
\centering
\includegraphics[width=0.725\columnwidth]{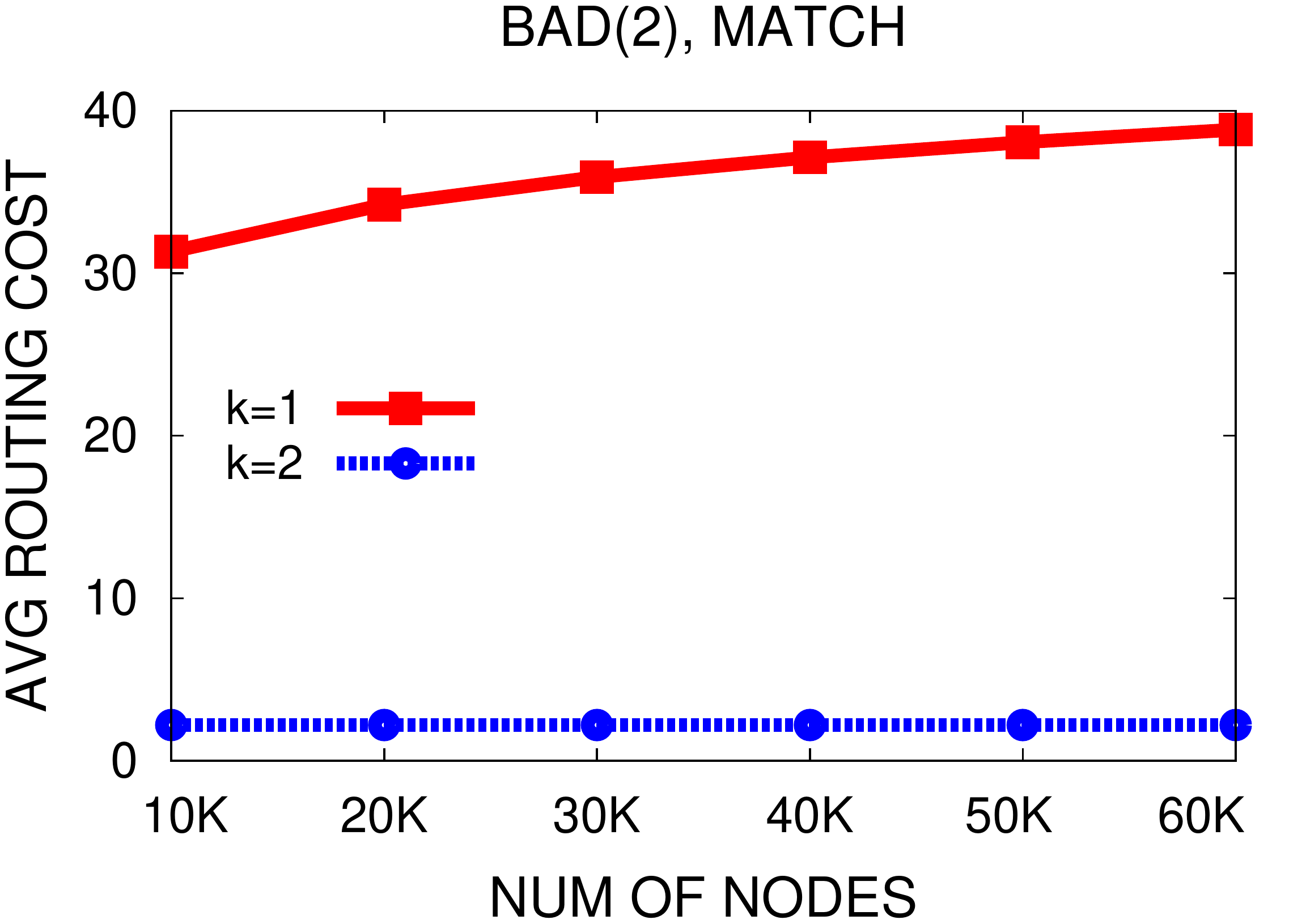}
\caption{Example where one additional BST can significantly improve costs. Scenario $\textsc{Bad}(2)$ with \textsc{Match} on \textsc{FB}.}
\label{fig:worst_requests}
\end{figure}

Let us now compare the alternative request sequences $\sigma$, generated using \textsc{Match}, \textsc{RW-0.5}, and \textsc{RW-1.0} on the Facebook graph $\mathcal{G}=\textsc{FB}$. Figure~\ref{fig:fb_match_and_rw} shows that for $\mathcal{H}=\Obst(1)$, the \textsc{Match} pattern yields the highest cost, but also improves the most for increasing $k$. The \textsc{RW-1.0} generally gives lower costs and as expected, \textsc{RW-0.5} reduces the costs further due to the temporal locality of the communications. Again, all $\Obst(k)$ overlays benefit from higher $k$ values.
\begin{figure}
\centering
\includegraphics[width=0.725\columnwidth]{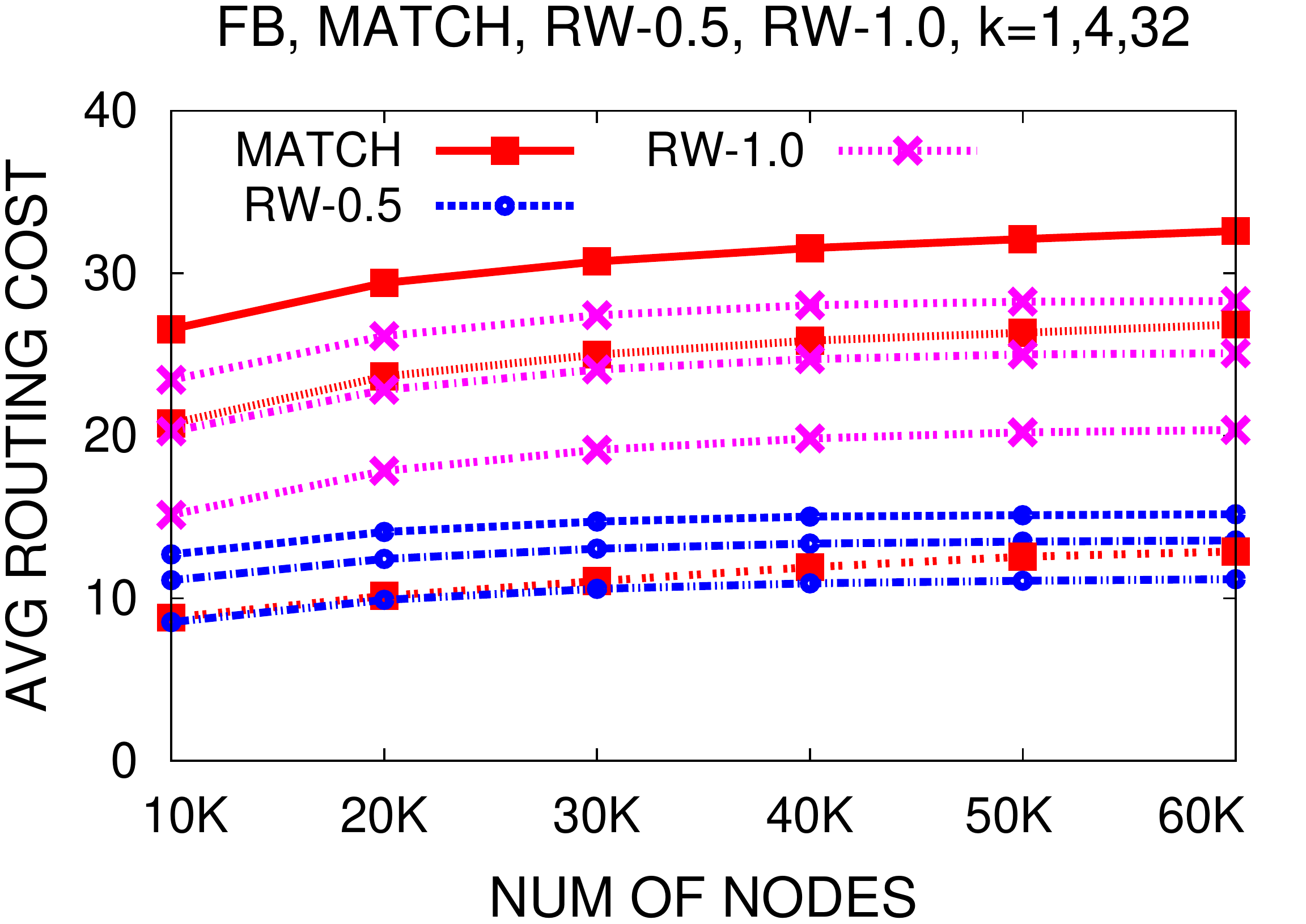}
\caption{Comparison of costs under \textsc{Match}, \textsc{RW-0.5}, and \textsc{RW-1.0} patterns on \textsc{FB} graph. The highest curve of every pattern is for $k=1$, the middle for $k=4$, and the lowest for $k=32$.}
\label{fig:fb_match_and_rw}
\end{figure}

Finally, Figure~\ref{fig:box_plots} studies the evolution of classical topology metrics over time, namely the min edge cut
and the diameter. In general, we observe that $\Obst(k)$ is relatively stable and behaves well also regarding these properties
and even for small $k$.
\begin{figure}
\centering
\includegraphics[width=0.725\columnwidth]{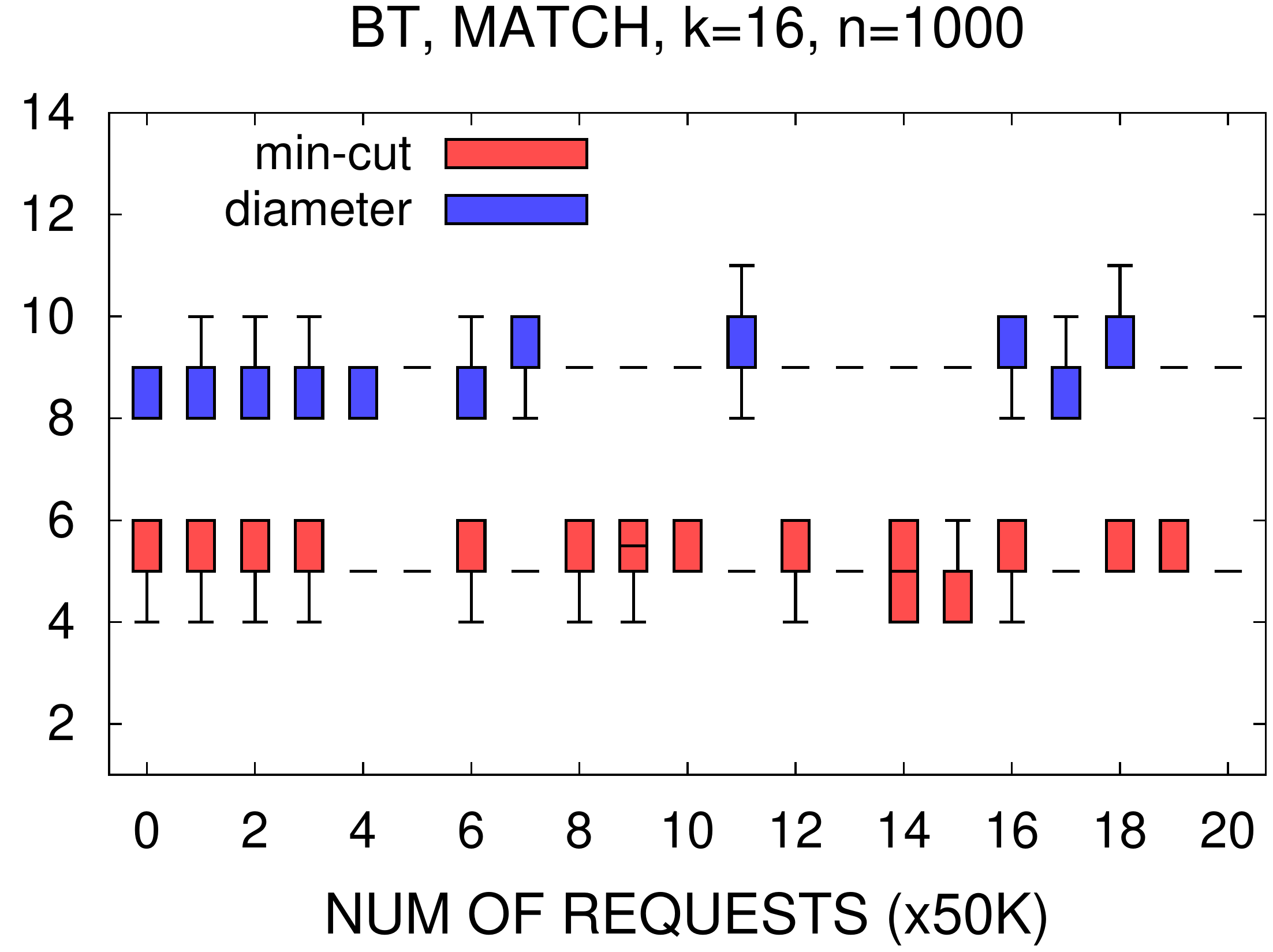}
\caption{Boxplot of diameter and mincut for 1000 peers and over time. (Ten repetitions.)}
\label{fig:box_plots}
\end{figure}
As shown in Figure~\ref{fig:fb_diam_over_time}, the diameter also scales well in the number of peers, although under \textsc{Match} it is slightly
larger than under \textsc{RW} and subject to more variance.
\begin{figure}
\centering
\subfloat[]{\includegraphics[width=0.725\columnwidth]{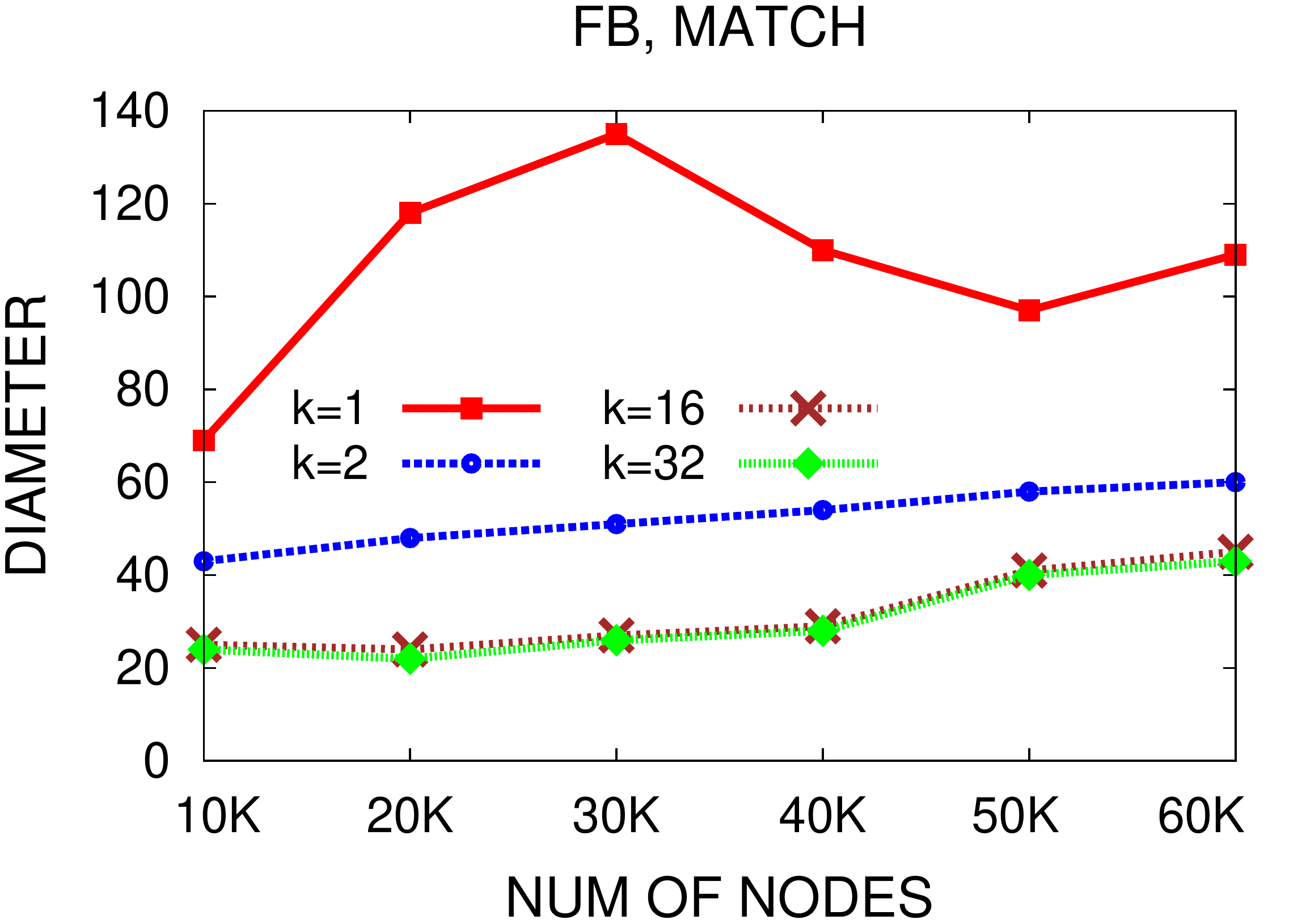}}\\
\subfloat[]{\includegraphics[width=0.725\columnwidth]{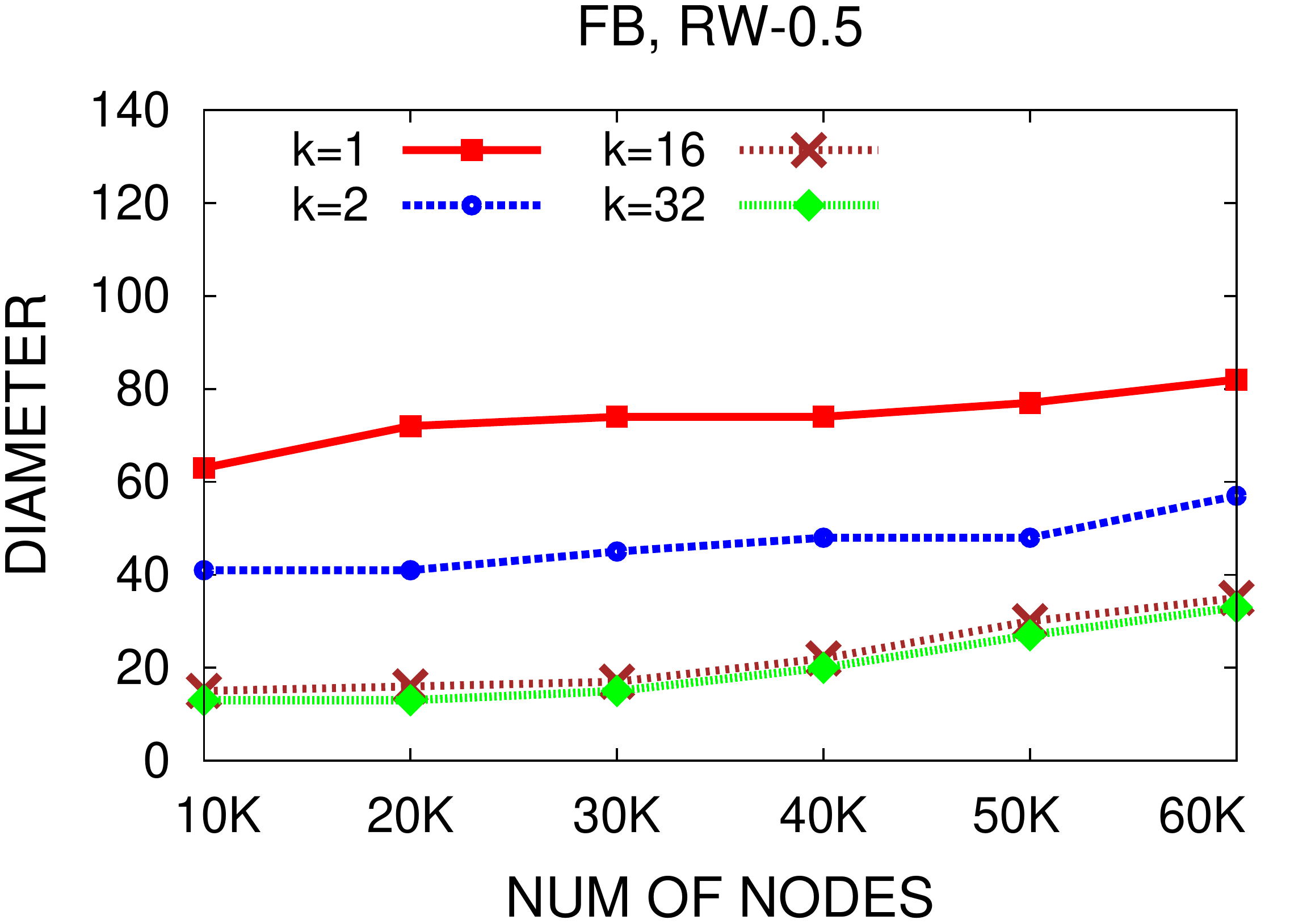}}
\caption{Diameter as a function of network size and for \textsc{Match} and \textsc{RW} pattern.}
\label{fig:fb_diam_over_time}
\end{figure}

However, let us emphasize that $\Obst(k)$ is optimized for amortized routing costs rather than mincut and diameter,
and we suggest using an additional, secondary overlay (e.g., a hypercubic topology) if these criteria are important.

\subsection{Convergence and Robustness}

Initializing $\Obst$ trees at random typically yields relatively low costs from the beginning.
Figure~\ref{fig:cost_plots} shows that the overlay subsequently also adjusts relatively quickly to
the specific demand. (Other scenarios yield similar results.) This indicates that the system
is able to adapt to new communication patterns and/or joins and leaves relatively quickly.

\begin{figure}
\centering
\includegraphics[width=0.725\columnwidth]{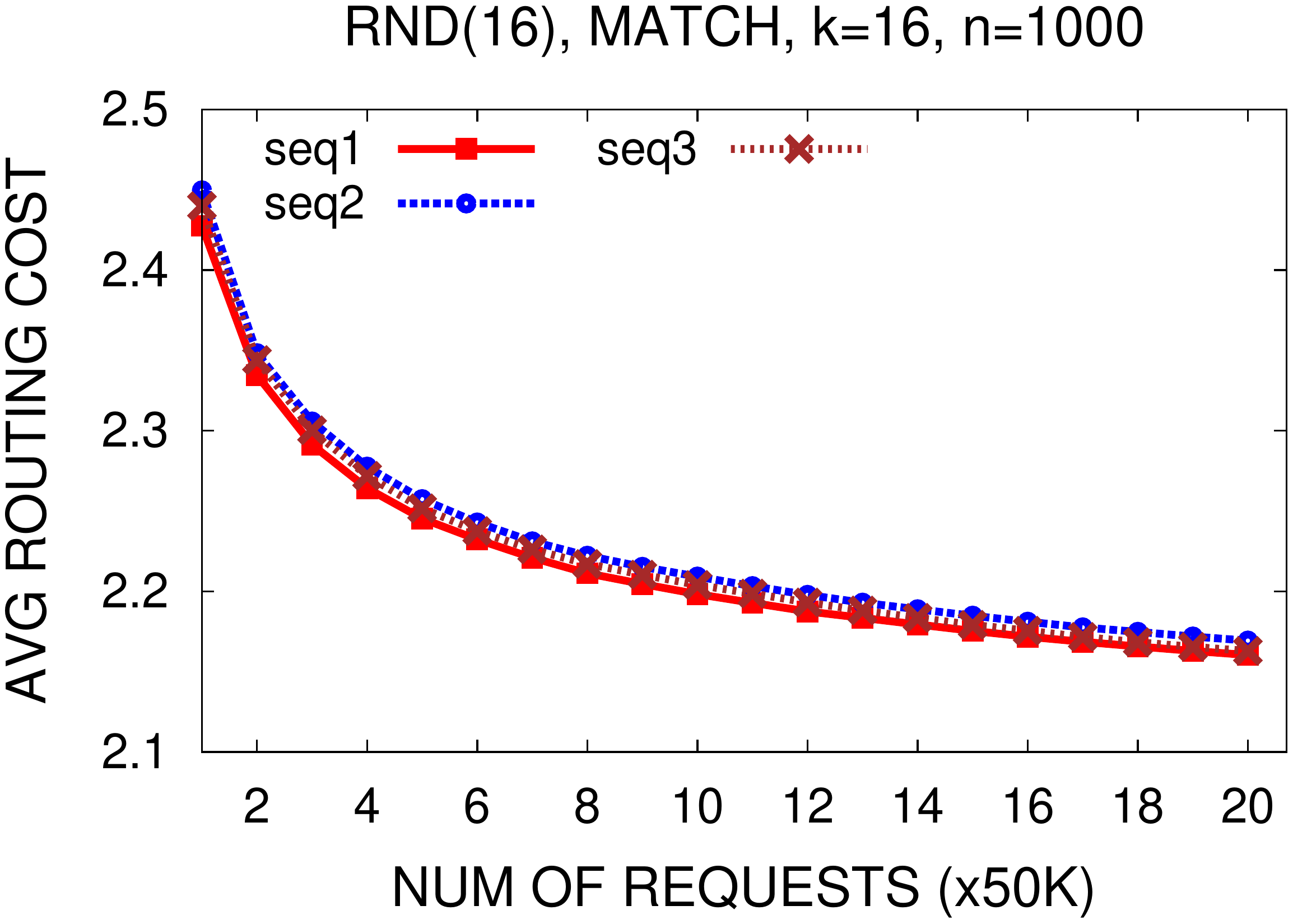}
\caption{Convergence for three representative runs under the $\Obst$ request model and for 1000 peers.}
\label{fig:cost_plots}
\end{figure}

The BSTs in $\Obst$ are also relatively independent in the sense that different links are used. We performed an experiment in which we run a long sequence of routing requests on the $\Obst(16)$ and then started to remove random peers. After each peer removal we measured the \emph{largest connected component} and the fraction of peer pairs that can communicate \emph{within a single connected BST}, and with respect to the overall $\Obst(16)$ graph connectivity. Single tree connectivity assumes that we only allow for local routing over a given BST, while the overall graph connectivity assumes that we can route on all the graph edges.

Figure~\ref{fig:robustness} shows that a large fraction of peers indeed stays connected by
a single BST, even under a large number of peer removals.
But the figure also shows that
if the connectivity is not counted with respect to a single BST (``Tree'' curve in the plot) but over all BSTs (``Graph'' curve in the plot), not surprisingly,
the robustness would be much higher. 

Our current $\Obst$ overlay employs
routing on single BSTs only. In order to exploit inter-BSTs links, either alternative routing
protocols (for instance a standard distance vector solution) could be used, or one may locally
recompute BSTs within the connected component. We leave these directions for future research,
and conclude that there is a potential for improvement in such scenarios.

\begin{figure}
\centering
\subfloat[]{\includegraphics[width=0.725\columnwidth]{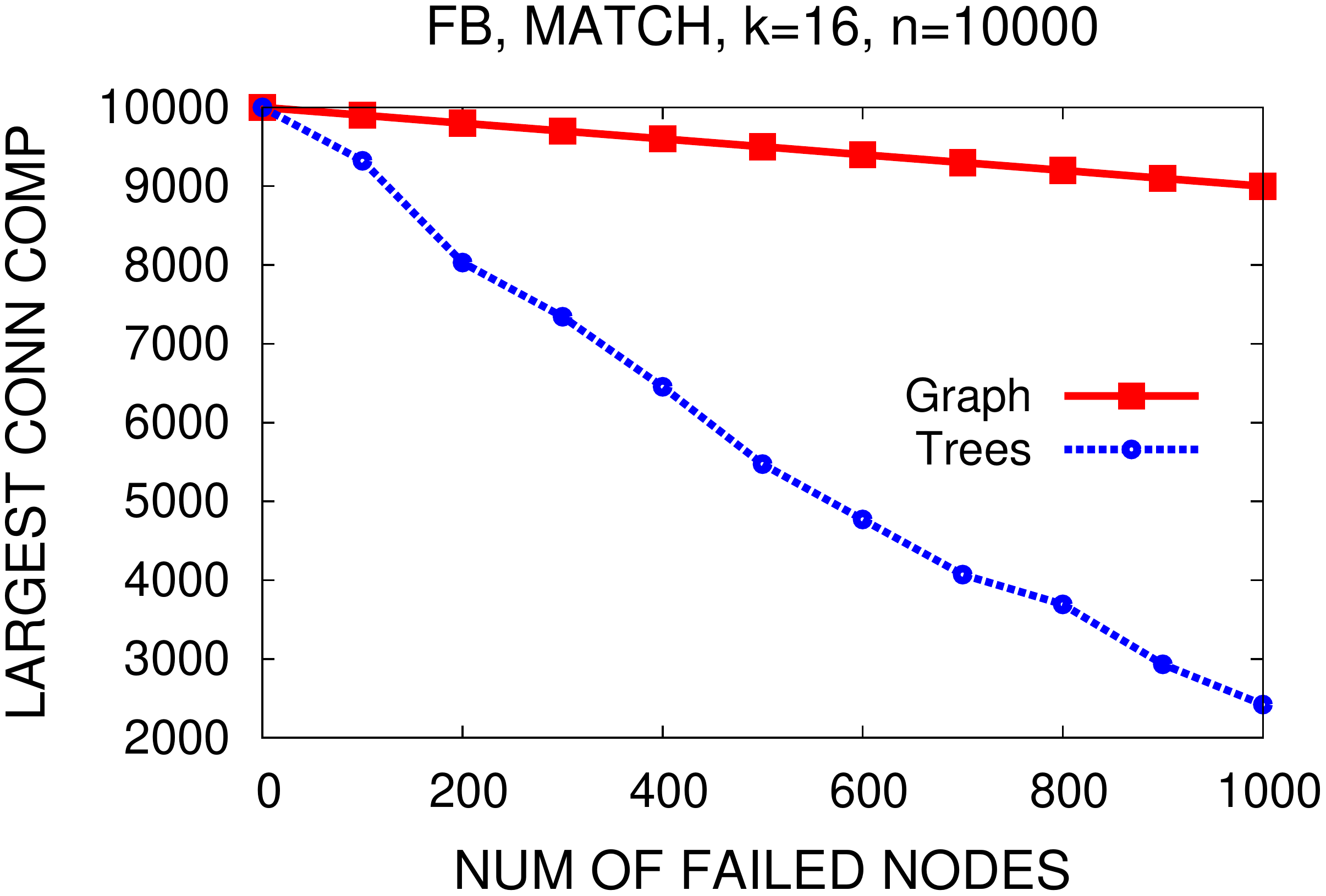}}\\
\subfloat[]{\includegraphics[width=0.725\columnwidth]{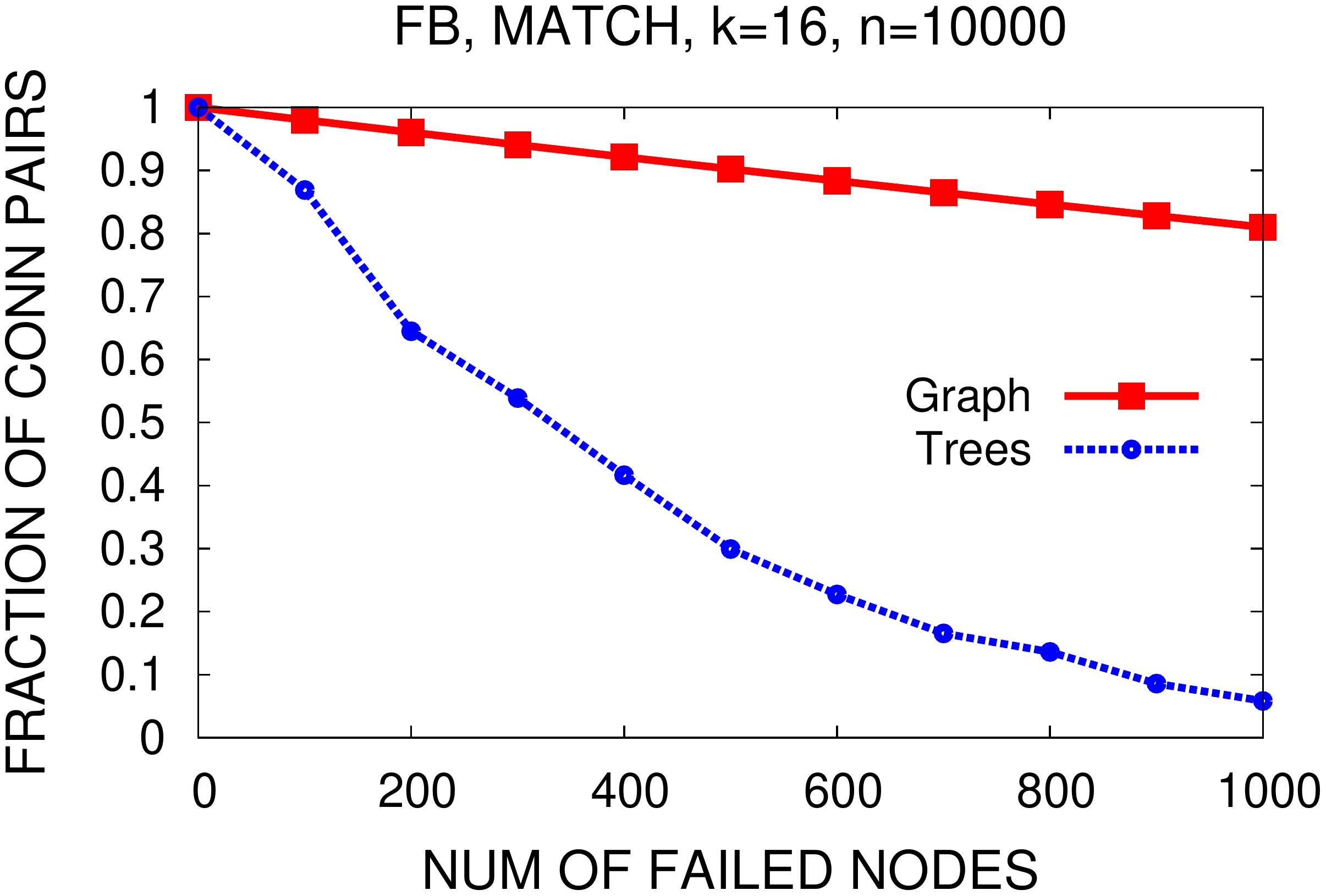}}
\caption{Largest connected component and fraction of connected pairs as a function of the number of randomly removed (failed) peers. ``Trees'' indicates the connectivity
with respect to a single BST, ``Graph'' indicates the connectivity with respect to all BST edges together.}
\label{fig:robustness}
\end{figure}

\subsection{Performance under Churn}\label{ssec:churn}

Peer-to-peer systems are typically highly dynamic, especially \emph{open} p2p networks where users can join and leave arbitrarily. This raises the question
of how well $\Obst$ can be tailored towards a traffic pattern under churn: How does the routing cost deteriorate with higher churn rates?
To investigate this question, we consider a simple scenario where after every routing request (in $\sigma$), $\lambda$ many random peers
leave $\Obst$ (i.e., are removed from all BSTs, according to the procedure sketched in Section~\ref{ssec:joinleave}). In order to keep the same dimension of the traffic matrix, upon the removal of each of the $\lambda$ peers, we immediately join another, new peer at the leaf of a corresponding BST.

Figure~\ref{fig:churn} shows how the routing cost depends on the churn rate $\lambda$. We can see that for increasing churn rates (and under traffic patterns generated using the $\textsc{Rnd}(16)$ guest graph), the routing cost increases moderately; for larger $\lambda$, the marginal cost effect becomes smaller. In additional experiments, we also observed that for the $\textsc{FB}$ and $\textsc{BT}$ guest graphs, the cost is almost agnostic to $\lambda$, implying that if $\Obst(k)$ is not perfectly adaptable to the traffic pattern, nodes can join and leave without affecting the routing cost by much.

\begin{figure}
\centering
\includegraphics[width=0.725\columnwidth]{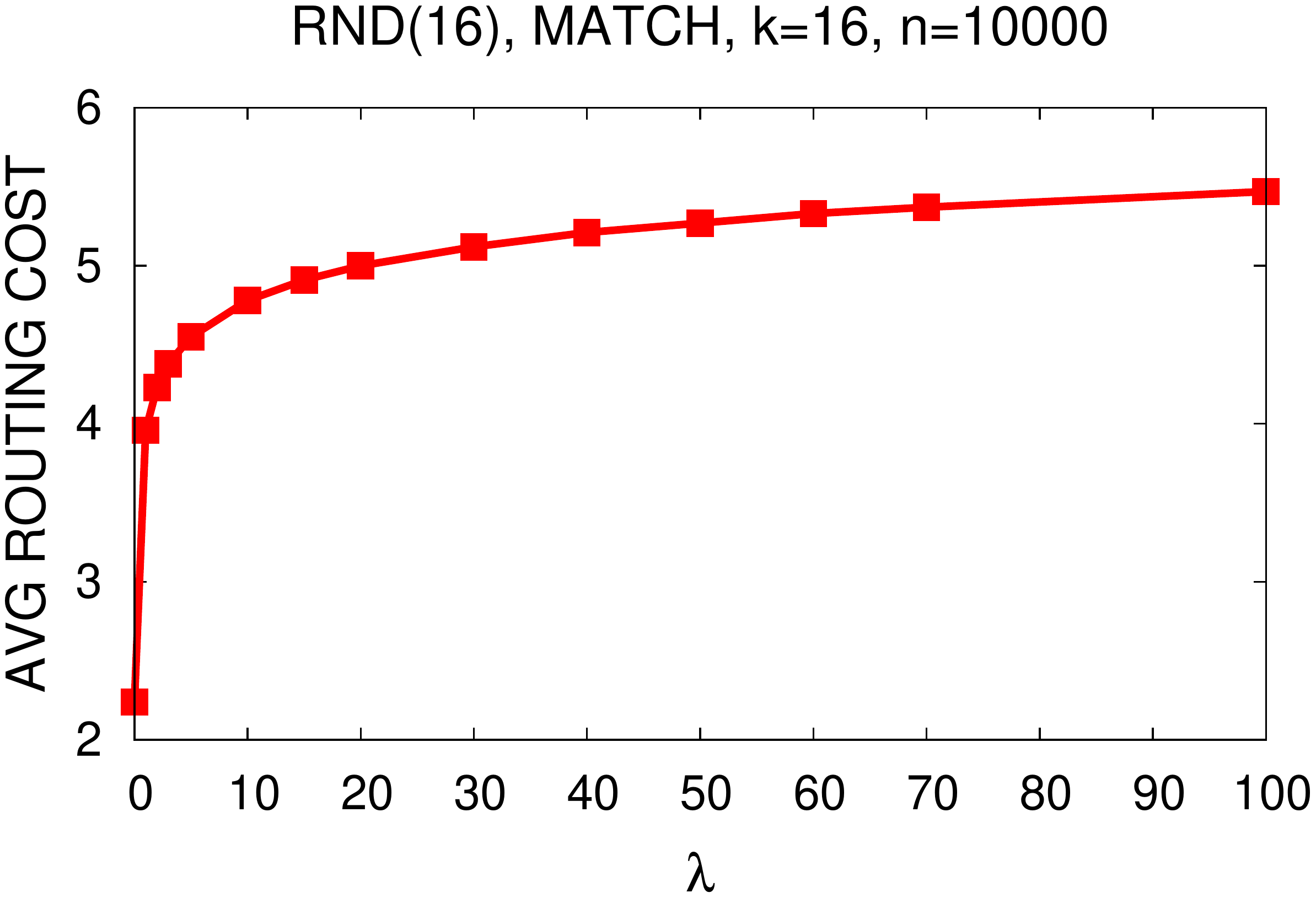}
\caption{Average routing distance when $\lambda$ nodes leave and rejoin between every routing request.}
\label{fig:churn}
\end{figure}

\section{Implications for a Lookup Model}\label{sec:discussion}

Interestingly, it turns out that while having multiple BSTs can significantly improve the routing cost (see Theorem~\ref{thm:2trees-bad}), the benefit of having parallel BSTs is rather limited in the context of classical lookup data structures, i.e., if all requests originate from a single node (the root).

Consider a sequence $\sigma=(v_0,v_1,\ldots,v_{m-1})$, $v_i\in V$ of \emph{lookup} requests, and $\left|V\right|=n$.
Theorem~\ref{static_lookup_lower_single} can be generalized to $k$ parallel lookup BSTs.
\begin{theorem}\label{static_lookup_lower_k_set}
Given $\sigma$, for any $\Obst(k)$:
\begin{align}
\Cost( \STAT, \Obst(k), \sigma) \ge \frac{H(\hat{Y})-\log k}{\log 3},
\end{align}
\noindent where $\hat{Y}(\sigma)$ is the empirical frequency distribution of $\sigma$ and  $H(\hat{Y})$ is its empirical entropy.
\end{theorem}
\begin{proof}
Let $f(i)$ denote the number of times the node $v_i$ appeared in the lookup sequence $\sigma$. The empirical frequency distribution is $\hat{p}(i)=f(i)/m$
for all $i$, and the entropy is given by $H(\hat{Y})=H(\hat{p}(1),\ldots,\hat{p}(n))$.
Since it is sufficient to serve a node by
one BST only, we can assume w.l.o.g.~each BST $T_i\in \Obst(k)$ is used to serve the lookup requests for a specific subset $V_i\in V$, and that $\bigcap_{i=1}^k V_i = \emptyset$ and $\bigcup_{i=1}^k V_i = V$.

Let $\hat{Y}_i$ be the empirical measure of the frequency distribution of the nodes in $V_i$ with respect to the lookup sequence $\sigma$.
Using entropy decomposition property, we can write:
$
H(\hat{Y})=H(\hat{p_1},\hat{p_2},\ldots,\hat{p}_n)$
$=H(\alpha_1,\alpha_2,\ldots,\alpha_k)+\sum_{i=1}^k\alpha_i H(\hat{Y}_i)$,
where $\alpha_i=\sum_{v_j\in V_i}\hat{p}_j$.

Always performing lookups on the optimal BST, we get using Theorem \ref{static_lookup_lower_single}:
$
\Cost( \STAT, \Obst(k), \sigma) \ge$ $ \sum_{i=1}^k \alpha_i \frac{1}{\log 3} H(\hat{Y}_i)$
$=\frac{H(\hat{Y})-H(\alpha_1,\alpha_2,\ldots,\alpha_k)}{\log 3}$
$\ge\frac{H(\hat{Y})-\log k}{\log 3}$.
\end{proof}

%
%



\section{Related Work}\label{sec:relwork}

The high rate of peer joins and leaves is arguably one of the
unique challenges of open p2p networks.
In order to deal with such transient behavior or even
topological attacks, many robust and self-repairing overlay networks
have been proposed in the literature:~\cite{skipplus,dynhc,shell}.
However, much less is known on networks which automatically optimize
towards a changing \emph{communication pattern}.

The p2p topologies studied in the literature are often hypercubic (e.g., Chord, Kademlia, or Skip Graph~\cite{aspnes2007skip}),
but there already exist multi-tree approaches, especially in the context of multicast~\cite{scribe} and streaming systems~\cite{disc07stefan}.

We are only aware of two papers on demand-optimized or self-adjusting overlay networks: Leitao et al.~\cite{xbot} study an overlay supporting gossip or epidemics on a dynamic topology. In contrast to our work, their focus is on unstructured networks (e.g., lookup or routing is not supported), and there is no formal evaluation. The paper closest to ours is~\cite{ipdps13stefan}. Avin et al.~initiate the study
of self-adjusting splay BSTs and introduce the double-splay algorithm. Although their work
regards a distributed scenario, it focuses on a single BST only. Our work builds upon these
  results and investigates the benefits of having multiple trees, which is also more realistic
  in the context of p2p computing.

More generally, one may also regard geography~\cite{geodemlia} or latency-aware~\cite{pastry} p2p systems as providing a certain degree of self-adaptiveness. However, these systems are typically optimized towards more static criteria, and change less frequently. This also holds for the p2p topologies tailored towards the ISPs' infrastructures~\cite{p2p-isp}.

Our work builds upon classic data
structure literature, and in particular on the splay
tree concept~\cite{sleator1985self}. Splay trees are optimized BSTs which move more popular items closer to the root in
order to reduce the average access time. Regarding
the splay trees as a network, \cite{sleator1985self} describes self-adjusting networks
for \emph{lookup} sequences, i.e., where the source is a \emph{single (virtual) node}
that is connected to the root. Splay trees have been studied
intensively for many years (e.g.~\cite{splaytrees2,sleator1985self}),
and the famous dynamic optimality conjecture continues to puzzle researchers~\cite{demaine2004dynamic}:
The conjecture claims that splay trees perform
as well as any other binary search tree algorithm.
Recently, the concurrent splay tree variant CBTrees~\cite{cbtree} has been proposed.
Unlike splay trees, CBTrees perform rotations infrequently and closer to the leaves;
this improves scalability in multicore settings.

Regarding the static variant of our problem, our work is also related to network design problems (e.g.,~\cite{net-design}) and,
more specifically, graph embedding algorithms~\cite{layout-sur}, e.g.,
the \emph{Minimum Linear Arrangement} (MLA) problem, originally studied by Harper~\cite{harper} to design
error-correcting codes. From our perspective,
MLA can be
seen as an early form of a ``demand-optimized'' embedding on the \emph{line} (rather than the BST as in our case):
Given a set of communication pairs, the goal is to flexibly arrange
the nodes on the line network such that the average communication
distance is minimized. While there exist many interesting algorithms
for this problem already (e.g., with sublogarithmic approximation
ratios~\cite{FL07}), no non-trivial results are known about
distributed and local solutions, or solutions on the tree as presented here.

\section{Conclusion}\label{sec:conclusion}

This paper initiated the study of p2p overlays which are statically optimized for or adapt to
specific communication patterns. We understand our algorithms and bounds as a first step,
and believe that they open interesting directions for future research.
For example, it would be interesting to study the multi-splay overlay from the perspective of online algorithms:
While computing the competitive ratio achieved by classic
splay trees (for lookup) arguably constitutes one of the most exciting open questions in
Theoretical Computer Science~\cite{demaine2004dynamic}, our work shows that the routing variant of the problem is rather different in nature (e.g.,
results in much lower cost).
 Another interesting research direction regards
alternative overlay topologies: while we have focused on a natural BST approach,
other graph classes such as the frequently used hypercubic networks and skip graphs~\cite{aspnes2007skip} may
also be made self-adjusting. Since these topologies also include tree-like subgraphs,
we believe that our results may serve as a basis for these extensions accordingly.


\textbf{Acknowledgments.} We are grateful to Chao Zhang, Prithula
Dhungel, Di Wu and Keith W.~Ross~\cite{bt-ross} for providing us with the
BitTorrent data.

\bibliographystyle{abbrv}

\begin{appendices}
\section{Proof of Theorem \ref{thm:2trees-bad}}
\begin{theorem6}\label{thm:2trees-bad2}
A single additional BST can reduce the amortized costs from a best possible value of $\Omega(\log{n})$ to $O(1)$.
\end{theorem6}
\begin{proof}
Consider the two BSTs $T_1=(V,E_1)$ and $T_2=(V,E_2)$ in Figure~\ref{fig:twotrees}.
Clearly, the two BSTs can be perfectly embedded into $\Obst(2)$ consisting of two BSTs as well. However, embedding the two trees at low cost
in one BST is hard, as we will show now.

Formally, we have that, where $V=\{1,\ldots,n\}$ (for an even $n$):
$E_1 =$ $(\{i,n/2-i+1\}: \forall i \in [1,n/4] )$ $\cup (\{n/2-i,i+2\}: \forall i \in [0,n/4-2] )$ $\cup (\{n/2+i,n-i\}: \forall i \in [0,n/4-1] )$
$\cup (\{n-i,n/2+1+i\}: \forall i \in [0,n/4-1] )$, and $E_2 = (\{i,n-i+1\}: \forall i \in [1,n/2] )$ $\cup (\{n-i,i+2\}: \forall i \in [0,n/2-2] )$
\noindent i.e., BST $T_2$ is ``laminated'' over the peer identifier space, and BST $T_1$ consists of two laminated subtrees over half of the nodes each.
Consider a request sequence $\sigma$ generated from these two trees with a uniform empirical distribution over all source-destination requests.
Clearly, optimal $\Obst(2)$ will serve all the requests with cost 2, since all the requests will be neighbors in $\Obst(2)$.
In order to show the logarithmic lower bound for the optimal $\Obst(1)$, we leverage the interval cut bound from Theorem~11 in~\cite{ipdps13stefan}. Concretely, we will show that for any interval $I$ of size $n/8 < \ell < n/4$  an $\Omega(\ell)$ (and hence an $\Omega(n)$) fraction of requests have one endpoint inside $I$ and the other endpoint outside $I$. In other words, each interval has a linear cut, and the claim follows since the empirical entropy is $\Omega(\log n)$.

The proof is by case analysis. \emph{Case~1}: Consider an interval $I=[x,x+\ell]$ where $x+\ell < n/2$. Then, the claim follows directly from tree $T_2$, as each node smaller or equal $n/2$ communicates with at least one node larger than $n/2$, so the cut is of size $\Omega(\ell)$. Similarly in \emph{Case~2} for an interval $I=[x,x+\ell]$ where $x \geq n/2$. In \emph{Case~3}, the interval $I=[x,x+\ell]$ crosses the node $n/2$, i.e., $x< n/2$ and $x+\ell > n/2$. Moreover, note that since $n/8 < \ell < n/4$, $n/4<x$ and $x+\ell<3n/4$ hold. The lower bound on the cut size then follows from tree $T_1$: each node $<n/2$ is connected to a node $\leq n/4$ outside the interval, and each node $>n/2$ is connected to a node $\geq n/2$ outside the interval.
\end{proof}


\section{Limitations on Perfect Overlays}
We examine in more detail the special case of \emph{perfect} overlays: overlays which accommodate a given communication pattern with the minimum
amortized costs of one. It is not surprising that many BSTs are required for perfect overlays, i.e., $k$ must be large in the $\Obst(k)$ overlay.

The concept of \emph{intersecting requests} plays a crucial role on the existence of perfect overlays.
\begin{definition}[Intersecting Requests ($\boldsymbol{\oplus}$)]\label{def:inter}
The communication requests between peer pair $(i,j)$ and peer pair $(k,\ell)$ \emph{intersect} if and only if $(k\in (i,j) \wedge \ell \notin (i,j)) \vee (k\notin (i,j) \wedge \ell \in (i,j))$. We denote an intersecting requests pair by $(i,j)\oplus(k,\ell)$.
\end{definition}
The inherent difficulty of perfectly embedding intersecting requests in a single BST is captured by the following lemma.
\begin{lemma}\label{lemma:intersected_req}
If $(i,j),(k,\ell)\in\sigma$ and $(i,j)\oplus(k,\ell)$, then if $(i,j)\in E(T)$, $(k,\ell)\notin E(T)$, for any BST $T$.
\end{lemma}
\begin{proof}
 W.l.o.g., assume $(i,j)\in E(T)$, $i>j$, $x\in (i,j)$, and that peer $j$ appears as a left child of peer $i$.
 Then, $x$ should be embedded into the right subtree of $j$ (since $j<k<i$). There are two cases: $\ell>i$ and $\ell<j$. If $\ell>i$, then $\ell$ cannot be in the left subtree of $i$, thus $(k,\ell)\notin E(T)$ (since $x$ is in the left subtree of $i$). If $\ell<j$, then $\ell$ cannot be in the right subtree of $j$, thus $(k,\ell)\notin E(T)$ (since $x$ is in the right subtree of $j$).
\end{proof}

The following theorem shows that already for seemingly independent communication requests due to a random matching on the set of peers,
many BSTs are required for a perfect $\Obst$.
\begin{theorem}\label{thm:nobody-is-perfect}
Let $\sigma$ be a sequence of communication requests coming from a random perfect matching on the complete peer graph. Then, with probability
at least $1-1/r$, there is no perfect overlay with
$k=\tfrac{n}{r}$ or less BSTs, for any $r\geq \sqrt{n\ln n}$ (a parameter).
\end{theorem}
\begin{proof}
We want to show that for $r\ge \sqrt{n\ln n}$, there exists, with high probability, a set of mutually intersecting requests of size $k=\tfrac{n}{r}$, and thus, we need at least $k$ trees to achieve perfect embedding (according to Lemma~\ref{lemma:intersected_req}).

Let us split all the $n$ nodes into $k=\tfrac{n}{r}$ consecutive non-overlapping intervals, each of size $r$: $\{1,\ldots,r\}\in I_1$, $\{r+1,\ldots,2r\}\in I_2$, \ldots, $\{n-r+1,\ldots,n\}\in I_k$. If there is at least one edge between a node in interval $i$ and a node in interval $j$, we say that these intervals are connected, and we denote this as: $I_i \leftrightarrow I_j$. A connection between a specific node $u$ to some node in the interval $I_i$, is denoted as $u \leftrightarrow I_j$. Now we find a bound on the probability that $I_i \nleftrightarrow I_j$. Let $u\in I_j$:
$
\Pr(u \leftrightarrow I_i)=\frac{r}{n-1}$, so
$\Pr(u \nleftrightarrow I_i)=1-\frac{r}{n-1}$, and
$\Pr(I_j \nleftrightarrow I_i) = \Pr\left(\bigcap_{u\in I_j} (u \nleftrightarrow I_i)\right)\le \left(1-\frac{r}{n-1}\right)^r$.

The last inequality is true since $\Pr\left(u \nleftrightarrow I_i \mid v \nleftrightarrow I_i\right)\le \Pr\left(u \nleftrightarrow I_i\right)$.

To continue the proof we need the following claim.
\begin{claim}
Let $r\ge\sqrt{n\ln n}$, then $(1-\frac{r}{n})^r\le \frac{1}{n}$.
\end{claim}
\begin{proof}
Let us denote $z=n\left(1-\frac{r}{n}\right)^r$. Then:
$\ln z = \ln n + r\ln{\left(1-\frac{r}{n}\right)}.
$ Using Taylor's expansion, we get:
$\ln{\left(1-\frac{r}{n}\right)}\le -\frac{r}{n}.
$
So:
$\ln z \le \ln n - \frac{r^2}{n}.
$ Thus, for $r\ge\sqrt{n\ln n}$, we get $\ln z\le 0$ and hence: $z\le 1$ which gives us:
$n(1-\frac{r}{n})^r  \le 1$, so
$(1-\frac{r}{n})^r \le \frac{1}{n}$.
\end{proof}

We can bound on the probability that two specific intervals are not connected:
$\Pr(I_j \nleftrightarrow I_i) \le \left(1-\frac{r}{n-1}\right)^r$
$\le \left(1-\frac{r}{n}\right)^r$
$\le \frac{1}{n}$.

Assuming that the number of intervals $x$ is even, consider the following scenario $\mathcal{S}$: $I_1 \leftrightarrow I_{x/2+1},I_2 \leftrightarrow I_{x/2+2},\ldots, I_{x/2} \leftrightarrow I_{x}$. Clearly, in this scenario, we have at least $x$ mutually intersecting requests. We can compute the probability for this to happen as:
$\Pr(\mathcal{S})=1-\Pr\left(\bigcup_{i=1,\ldots,k/2}(I_i \nleftrightarrow I_{i+k/2})\right)$
$\ge 1 - k\cdot\frac{1}{n}$
$=1-\frac{1}{r}$.
\end{proof}

\end{appendices}

\end{document}